\documentclass[11pt,letterpaper]{article}

\usepackage[utf8]{inputenc}
\usepackage[T1]{fontenc}
\usepackage[english]{babel}

\usepackage{amsmath,amssymb,amsthm}
\usepackage{natbib}
\usepackage[margin=1in]{geometry}
\usepackage{setspace}
\usepackage{graphicx}
\usepackage{url}
\usepackage[hidelinks]{hyperref}
\usepackage{xspace}
\usepackage{indentfirst}
\usepackage{graphics}
\usepackage{tikz}
\usetikzlibrary{shapes.geometric, arrows, positioning, decorations.pathreplacing, shadows}
\usepackage{caption}
\usepackage{algorithm2e}
\usepackage{booktabs}
\usepackage{multirow}
\usepackage{arydshln}
\usepackage{bbm}
\usepackage{array}
\usepackage{mdframed}
\usepackage{xcolor}

\newcolumntype{L}[1]{>{\raggedright\let\newline\\\arraybackslash\hspace{0pt}}m{#1}}
\newcolumntype{C}[1]{>{\centering\let\newline\\\arraybackslash\hspace{0pt}}m{#1}}
\newcolumntype{R}[1]{>{\raggedleft\let\newline\\\arraybackslash\hspace{0pt}}m{#1}}

\SetAlFnt{\small}
\SetAlCapFnt{\normalsize}
\SetAlCapNameFnt{\normalsize}
\SetNlSty{textbf}{}{}

\newcommand{\br}[1]{\mathopen{}\left( #1 \right)}
\newcommand{\brc}[1]{\mathopen{}\left\{ #1 \right\}}
\newcommand{\spr}[1]{\mathopen{}\left| #1 \right|}

\newcommand{\cl}[1]{\mathopen{}\left\lceil #1 \right\rceil}
\newcommand{\angl}[1]{\mathopen{}\langle #1 \rangle}

\newcommand{\cov}{\operatorname{cov}}
\newcommand{\prefix}[2]{\operatorname{prefix}\br{#1,#2}}
\newcommand{\pref}[2]{\prefix{#1}{#2}}

\newcommand{\NPhard}{\textnormal{NP-Hard}\xspace}

\newcommand{\OPT}{\text{OPT}}
\newcommand{\COST}{\text{COST}}
\newcommand{\COSTW}{\text{COST}_W}
\newcommand{\COSTA}{\text{COST}_A}

\newcommand{\argmin}{\mathopen{}\operatorname*{arg\,min}}
\newcommand{\argmax}{\mathopen{}\operatorname*{arg\,max}}

\newcommand{\cH}{\mathcal{H}}

\newcommand{\cS}{\mathcal{S}}
\newcommand{\cT}{\mathcal{T}}
\newcommand{\cC}{\mathcal{C}}
\newcommand{\cF}{\mathcal{F}}
\newcommand{\cU}{\mathcal{U}}
\newcommand{\cI}{\mathcal{I}}

\newcommand{\cQ}{\mathcal{Q}}
\newcommand{\cA}{\mathcal{A}}

\newcommand{\cO}{\mathcal{O}}
\newcommand{\cX}{\mathcal{X}}
\newcommand{\questioner}{questioner }
\newcommand{\target}{h^*}
\newcommand{\bigo}{\mathcal{O}}

\newcommand{\ProblemDT}{\textsc{DT}}
\newcommand{\ProblemPCWCDT}{\textsc{PCWCDT}}
\newcommand{\ProblemPCACDT}{\textsc{PCACDT}}
\newcommand{\ProblemPCSC}{\textsc{f-PCSC}}
\newcommand{\ProblemBPCSC}{\textsc{b-PCSC}}
\newcommand{\ProblemPCMSSC}{\textsc{f-PCMSSC}}
\newcommand{\ProblemMDPCS}{\textsc{MDPCS}}
\newcommand{\ProblemBMDPCS}{\textsc{b-MDPCS}}
\newcommand{\ProblemGST}{\textsc{GST}}
\newcommand{\ProblemPDS}{\textsc{PDS}}
\newcommand{\ProblemGSO}{\textsc{GSO}}
\newcommand{\ProblemLPGST}{\textsc{LPGST}}

\newcommand{\ProcDecisionTree}{\textsc{DecisionTree}{} }
\newcommand{\blackBox}{\textsc{FindCover}{} }
\newcommand{\ProcedurePCSC}{\textsc{Greedy-PCSC}}
\newcommand{\ProcedureBPCSC}{\textsc{Greedy-b-PCSC}}
\newcommand{\ProcedureFPCMSSC}{\textsc{Procedure-f-PCMSSC}}

\newcommand{\optPCSC}[1]{\OPT\br{#1}}

\newcommand{\cCopt}{\cC^{*}} %optimal solution to PCSC

\newcommand{\tests}[2]{\operatorname{tests}(#1,#2)} %tests done in dec. tree #1 to arrive at node #2
\newcommand{\paraTitle}[1]{\noindent\textbf{#1}}   %paragraph titles (to eliminate more sub-sub-subsections
\newcommand{\sepcover}[1]{\cS^*\br{#1}}
\newcommand{\sepcoverD}[1]{\cS\br{#1}}
\newcommand{\coverageTime}{\operatorname{ct}}  %total coverage time of a sequence for PCMSSC
\newcommand{\closure}[1]{\operatorname{closure}[#1]}
\newcommand{\budget}{b}

\newtheorem{theorem}{Theorem}[section]
\newtheorem{lemma}[theorem]{Lemma}

\newtheorem{corollary}[theorem]{Corollary}
\newtheorem{conjecture}[theorem]{Conjecture}
\theoremstyle{definition}
\newtheorem{definition}[theorem]{Definition}
\newtheorem{problem}[theorem]{Problem}
\newtheorem{claim}[theorem]{Claim}
\theoremstyle{plain}
\newtheorem{observation}[theorem]{Observation}
\theoremstyle{remark}

\title{Precedence-Constrained Decision Trees and Coverings}
\author{%
Michał Szyfelbein\\
Gdańsk University of Technology, 80-233 Gdańsk, Poland\\
\texttt{michal.szyfelbein@pg.edu.pl}
\and
Dariusz Dereniowski\\
Gdańsk University of Technology, 80-233 Gdańsk, Poland\\
\texttt{deren@eti.pg.edu.pl}
}
\date{}

\begin{document}

\maketitle

\begin{abstract}
This work considers a number of optimization problems and reductive relations between them.
The two main problems we are interested in are the \textsc{Optimal Decision Tree} and \textsc{Set Cover}.
We study these two fundamental tasks under precedence constraints, that is, if a test (or set) $X$ is a predecessor of $Y$, then in any feasible decision tree $X$ needs to be an ancestor of $Y$  (or respectively, if $Y$ is added to set cover, then so must be $X$).
For the \textsc{Optimal Decision Tree} we consider two optimization criteria: worst case identification time (height of the tree) or the average identification time.
Similarly, for the \textsc{Set Cover} we study two cost measures: the size of the cover or the average cover time.

Our approach is to develop a number of algorithmic reductions, where an approximation algorithm for one problem provides an approximation for another via a black-box usage of a procedure for the former.
En route we introduce other optimization problems either to complete the `reduction landscape' or because they hold the essence of combinatorial structure of our problems.
The latter is brought by a problem of finding a \textsc{Maximum Density Precedence-Closed Subfamily}, where the density is defined as the ratio of the number of items the family covers to its size. We provide $\mathcal{O}^*(\sqrt{m})$-approximation polynomial-time algorithms for all aforementioned problems.
The picture is complemented by a number of hardness reductions that provide $\mathcal{O}(m^{1/12-\epsilon})$-inapproximability results for the decision tree and covering problems.
Besides giving a complete set of results for general precedence constraints, we also provide polylogarithmic approximation guarantees for two most typically studied and applicable graph types, outforests and inforests. By providing corresponding hardness results, we show most of these results to be tight.
\end{abstract}


\section{Introduction}

Consider a set $\cH$ of $n$ \emph{hypotheses}, a set $\cT$ of $m$ \emph{tests} and  an unknown \emph{target hypothesis} $\target\in\cH$ that needs to be discovered through testing.
Each test $t\in\cT$ is a partition of $\cH$, that is, $t$ consists of subsets of $\cH$ such that for any $X,Y\in t$,  $X\cap Y=\emptyset$ and $\bigcup_{X\in t} X=\cH$.
As a result of executing a test $t\in\cT$, the \questioner receives a \emph{reply} that reveals $H\in t$ such that $\target\in H$. Upon receiving this response, the \questioner adaptively selects the next test from $\cT$ to perform, until the target hypothesis is identified. The goal is to conceive a strategy of testing, often modeled as a decision tree, which either minimizes the worst-case or the average-case number of tests performed until $\target$ is uncovered. We refer to the former problem as the \textsc{Worst Case Decision Tree} and to the latter as the \textsc{Average Case Decision Tree}. If the cost criterion is not mentioned explicitly, we refer to the problem as the \textsc{Optimal Decision Tree} ($\ProblemDT$).

In this work we generalize the above setup by considering instances which are subject to arbitrary precedence constraints between tests given as a partial order $(\cT,\preceq)$ on the test set $\cT$.
It is required that in any valid strategy of testing, a given test $t\in \cT$ can be performed only after all its predecessors have been performed previously (to which we refer as \emph{precedence-closed} solution). This setup is a natural generalization of the standard \textsc{Optimal Decision Tree}, and is partially motivated by the fact that the decision tree can be interpreted as a schedule of tests. Since task scheduling often involves precedence constraints, it is natural to ask how such constraints affect the complexity of the \textsc{Optimal Decision Tree}.
On the applicability side, \textsc{Optimal Decision Tree} includes numerous use cases, some of them in the area of medicine and biology \cite{BayesianLearnerOptimalNoisyBinarySearch}.
For example hypotheses may represent possible conditions or diseases and tests may represent medical procedures. Such a setup suggest inclusion of precedence constraints in various forms. For instance, some tests may be more invasive or expensive than others, and should be performed only after cheaper or less invasive tests have been performed. Additionally, it can also happen that a given test is a prerequisite for another test, e.g., a blood test may be required before performing a biopsy. At last it might also be the case that a large test consists of multiple stages, each with its own outcomes, and the later stages can be performed only after the earlier ones.

% \begin{figure}[t!]
% \begin{minipage}[t]{0.47\textwidth}
% \input{figures/pcal_example}
% \caption{Decision tree with precedence.}\label{fig:pcal_example}
% \end{minipage}
% \hfill
% \begin{minipage}[t]{0.47\textwidth}
% \input{figures/pccp_example}
% \caption{Set cover with precedence.}\label{fig:pccp_example}
% \end{minipage}
% \end{figure}

Algorithmically, the aforementioned precedence constraints make the problem significantly more challenging. In a classical setup (without such constraints) the usual way of proceeding is to select at each step a test that maximizes some measure of information gain (e.g., reduction in entropy or number of hypotheses pairs). However, when precedence constraints are present, it may happen that the most informative test cannot be chosen at the current step, as some of its predecessors have not been performed yet.
To tackle this inconvenience we use a series of `reductions' between optimization problems, i.~e., using an algorithm for one problem as a subroutine for the other.
Of particular interest to us becomes set covering with precedence constraints:
We are given a set $\cU$ of $n$ items, a collection $\cS$ of $m$ subsets of $\cU$, such that $\bigcup_{S\in\cS}S=\cU$, a partial order $(\cS,\preceq)$ on these subsets and a parameter $0<f\leq 1$.
We say that a subfamily $\cC\subseteq\cS$ \emph{covers} at least $f$-fraction of items from $\cU$ if $\spr{\bigcup_{C\in\cC}C}\geq f\cdot n$.
We ask for a $\cC\subseteq\cS$ that covers at least $f$-fraction of the items from $\cU$ and is \emph{precedence-closed}, that is, for each $x\in\cC$ and each $y\in\cS$ such that $y\preceq x$ it holds $y\in\cC$. We measure the quality of such $\cC$ by either its size $\spr{\cC}$ or the min-sum time it takes to cover an item from $\cU$ (assuming some order of sets in $\cC$ closed under the precedence constraint). Depending on the cost objective we refer to the problem as the \textsc{Precedence Constrained Set Cover} ($\ProblemPCSC$) or the \textsc{Precedence Constrained Min-Sum Set Cover} ($\ProblemPCMSSC$). 
% For a visual example of such a scenario consider Figure~\ref{fig:pccp_example}.
Our main approach is to show how to reduce the \textsc{Optimal Decision Tree} to set covering. Then we develop approximation algorithms for the set covering problems themselves. The approximation algorithms that we propose en route, e.g. for precedence constrained set covering, are of independent interest.

\subsection{Related Work}

\paragraph{Optimal Decision Tree}
\textsc{Optimal Decision Tree}, also referred to as \textsc{Binary Search Trees} in \cite{BoseCIKL20,DemaineHIP07}, is an extensively studied problem in computer science, starting with \cite{GareyPerfBoundsOnSplittingAlgForBinTesting} in the 1970s.
Since then it has gathered a lot of attention due to its numerous applications including medical diagnosis, troubleshooting, active learning, and information retrieval. 
Usually two optimization criteria are considered: the worst-case and the average-case cost.
Both versions of the problem are \NPhard and cannot be approximated within an $o\br{\log n}$ factor \cite{ApproximatingDecisionTreesMultiwayBranches,DiagnosisDetermination,ConstructOptimalBinaryDecisionTreesIsNPComplete,HardnessOfMinHeightDTP}.
Moreover, this bound is tight and several $\cO(\log n)$-approximation algorithms are known for the average-case both with non-uniform probabilities and test costs \cite{ApproximatingOptimalBinaryDecisionTrees,DTsforEntIdent,ApproximatingDecisionTreesMultiwayBranches,DiagnosisDetermination,MinimumCostAdaptiveSubmodularCover,AnalysisGreedyActiveLearning,AdaptivityInAdaptiveSubmodularity,AverageCaseActiveLearningWithCosts,ApproxAlgsForOptDTsAndAdapTSPProblems,OptimalSplitTreeProblem,AdaptiveSubmodularRankingAndRouting} as well as the worst case \cite{DecisionTreesForGeometricModels,TheCostComplexityOfInteractiveLearning,DiagnosisDetermination}. Few special cases are known to admit $o\br{\log n}$-approximation. For the average case, this includes an $\cO\br{\log n/\log \log n}$-approximation when tests have a constant number of possible outcomes and all probabilities and costs are uniform \cite{TightAnalysisGreedyUniformDecisionTree}.
% This variant cannot be approximated within $\br{4-\epsilon}$ factor for any $\epsilon > 0$ unless P = NP \cite{DTsforEntIdent}. 
Moreover, obtaining approximation ratio of $\br{9+\epsilon}$ can be done in subexponential time, and thus is not \NPhard assuming ETH \cite{TightAnalysisGreedyUniformDecisionTree}.
For the worst case, when the underlying search space is a partially ordered set with one maximum element and costs are uniform an $\cO\br{\log n/\log \log n}$-approximation is known \cite{EdgeRankingSearchingPartialOrders}.
To the best of our knowledge, the \textsc{Optimal Decision Tree} problem has not been studied in the presence of precedence constraints.

\paragraph{Searching in Trees}
The special case when input instance represents binary searching in a tree has also been extensively studied.
% In this problem hypothesis represent vertices and tests represent queries about direction towards the target. 
For uniform costs and the worst-case criterion, a linear time algorithms are known for both edge and vertex query variants \cite{Mozes_Onak2008FindOptTSStartInLinTime,OnakParys2006GenOfBSSInTsAndFLikePosets,Schaffer1989OptNodeRankOfTsInLinTime}. For the average case, uniform costs and vertex queries an FPTAS exists \cite{SearchTreesOnGraphs}. For edge queries the problem is known to be \NPhard \cite{OnTheComplexityOfSearchingInTreesAverageCaseMinimization} and a greedy strategy achieves $3/2$-approximation \cite{ImprovedApproximationAlgorithmsForTheAverageCaseTreeSearchingProblem,TightApproximationBoundsOnASimpleAlgorithmForMinimumAverageSearchTimeInTrees,OnTheComplexityOfSearchingInTreesAverageCaseMinimization}. For non-uniform costs, the vertex query model generalizes edge query model which is known to be \NPhard \cite{EdgeRankingOfWeightedTrees,TheBinaryIdentificationProblemForWeightedTrees,OnTheTreeSearchProblemWithNonUniformCosts}. The best known approximation ratio is $\bigo\br{\sqrt{\log n}}$ \cite{ApproximationStrategiesforGeneralizedBinarySearchinWeightedTrees}.
The average case is also \NPhard and an $\br{4+\epsilon}$-approximation FPTAS is known \cite{szyfelbein2025approximatingaveragecasegraphsearch}.
Searching in trees is a generalization of the classical binary search intensively studied in various models \cite{BayesianLearnerOptimalNoisyBinarySearch,BorgstromK93,DereniowskiLU25,LaberMP02}.
A generalization of tree search to general graphs has also been recently studied, mostly in presence of noise, where some replies can be erroneous \cite{DereniowskiLU25,Emamjomeh-Zadeh16}.
For some applications of these problems see e.g. \cite{Emamjomeh-Zadeh17,KarpK07}.

\paragraph{Set Cover with precedence constraints}

\textsc{Set Cover} is among the most important problems in combinatorial approximation algorithms.
It is well known that the greedy algorithm achieves an $H_n$-approximation for \textsc{Set Cover} \cite{GreedyHeuristicSetCoverProblem}, where $H_n=\cO\br{\log n}$ is the $n$-th harmonic number. This is tight since it cannot be approximated within a $(1-\epsilon)\ln n$ factor for any $\epsilon > 0$ unless P=NP \cite{AnalyticalApproachToParallelRepetition}.
The \textsc{Min-Sum Set Cover} is a variant of the problem in which the goal is to minimize the average cover time of elements. For this version, the greedy algorithm is $4$-approximate \cite{ApproximatingMinSumSetCover}, which is tight.
When allowing arbitrary precedence constraints the problem admits an $\cO\br{\sqrt{m}}$-approximation and cannot be approximated within an $\cO\br{m^{1/12-\epsilon}}$ nor $\cO\br{n^{1/6-\epsilon}}$ factor \cite{PCMSSC} subject to the $\hyperref[problem:PDS]{\ProblemPDS}$ \cite{OnApproxTargetSetSelection}.

\subsection{Our results and techniques} \label{sec:our-results}

% Our main contribution consists of approximation algorithms and hardness results for decision tree and set covering problems with precedence constraints, studied under different structural restrictions on the precedence relation.

\paragraph{Approximation Results.} The key insight underlying our approach is a systematic hierarchy of reductions: we reduce decision tree problems to set covering problems with precedence constraints, which in turn are solved via algorithms for finding dense precedence-closed subfamilies. This hierarchy of reductions, illustrated by Figure~\ref{fig:reductions}, allows us to iteratively simplify the problem until we reach some core task that can be approximated effectively.
All problems are formally stated in Section~\ref{sec:preliminaries} below.
Our approximation results, including best known results from the literature, are summarized in Table~\ref{tab:results}.

\begin{figure}[ht!]
\centering
\begin{tikzpicture}[
    node distance=0.7cm and 2cm,
    problem/.style={rectangle, draw, very thick, minimum width=1.8cm, minimum height=0.6cm, align=center, font=\scriptsize, fill=white, drop shadow, rounded corners=2pt},
    arrow/.style={->, very thick},
    edgelabel/.style={font=\tiny, fill=white, inner sep=1pt}
]

% Define positions based on the hand-drawn diagram
% Top level
\node[problem] (bmdpcs) {$\hyperref[problem:BMDPCS]{\ProblemBMDPCS}$};

% Second level - left and center
\node[problem, below left=0.7cm and 2.5cm of bmdpcs] (mdpcs) {$\hyperref[problem:MDPCS]{\ProblemMDPCS}$};
\node[problem, below=0.7cm of bmdpcs] (bpcsc) {$\hyperref[problem:BPCSC]{\ProblemBPCSC}$};

% Second level - right
\node[problem, above right=0cm and 1.5cm of bpcsc] (gso) {$\hyperref[problem:GSO]{\ProblemGSO}$};
\node[problem, below right=0.7cm and 1.5cm of gso] (lpgst) {$\hyperref[problem:LPGST]{\ProblemLPGST}$};

% Third level
\node[problem, below left=0.7cm and 0.8cm of bpcsc] (fpcsc) {$\hyperref[problem:PCSC]{\ProblemPCSC}$};
\node[problem, below right=0.7cm and 0.8cm of bpcsc] (fpcmssc) {$\hyperref[problem:PCMSSC]{\ProblemPCMSSC}$};

% Bottom level
\node[problem, below=0.7cm of fpcsc] (pcwcdt) {$\hyperref[problem:PCWCDT]{\ProblemPCWCDT}$};
\node[problem, below=0.7cm of fpcmssc] (pcacdt) {$\hyperref[problem:PCACDT]{\ProblemPCACDT}$};

% Arrows with theorem references
\draw[arrow] (bmdpcs) -- node[edgelabel, pos=0.5] {\hyperref[subsection:BPCSC]{Sec.~\ref*{subsection:BPCSC}}} (bpcsc);
\draw[arrow] (mdpcs) -- node[edgelabel, pos=0.5] {\hyperref[subsection:PCSC]{Sec.~\ref*{subsection:PCSC}}} (fpcsc);
\draw[arrow] (bpcsc) -- node[edgelabel, pos=0.5] {\hyperref[subsection:PCSC]{Sec.~\ref*{subsection:PCSC}}} (fpcsc);
\draw[arrow] (bpcsc) -- node[edgelabel, pos=0.5] {\hyperref[subsection:PCMSSC]{Sec.~\ref*{subsection:PCMSSC}}} (fpcmssc);
\draw[arrow] (gso) -- node[edgelabel, pos=0.5] {\hyperref[subsection:SpecialCases]{Sec.~\ref*{subsection:SpecialCases}}} (bpcsc);
\draw[arrow] (gso) -- node[edgelabel, pos=0.5] {\cite{ApproxAlgsForOptDTsAndAdapTSPProblems}} (lpgst);
\draw[arrow] (lpgst) -- node[edgelabel, pos=0.5] {\hyperref[subsection:SpecialCases]{Sec.~\ref*{subsection:SpecialCases}}} (fpcmssc);
\draw[arrow] (fpcsc) -- node[edgelabel, pos=0.5] {\hyperref[subsection:PCWCDT]{Sec.~\ref*{subsection:PCWCDT}}} (pcwcdt);
\draw[arrow] (fpcmssc) -- node[edgelabel, pos=0.5] {\hyperref[subsection:PCACDT]{Sec.~\ref*{subsection:PCACDT}}} (pcacdt);

\end{tikzpicture}
\caption{Relationships between covering and decision tree problems, $\Pi_1 \to \Pi_2$ denotes that an approximation algorithm for problem $\Pi_1$ implies an approximation algorithm for problem $\Pi_2$.}\label{fig:reductions}
\end{figure}
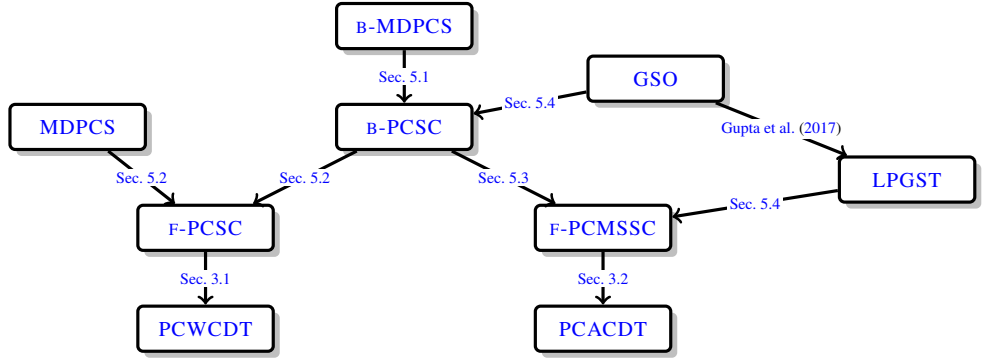

\begin{table}[ht!]
\caption{Approximation guarantees for covering and decision tree problems under different precedence constraints. ($H_n = \sum_{i=1}^{n} \frac{1}{i} = \Theta(\log n)$ is the $n$-th \emph{harmonic number}.).}
\label{tab:results}
\centering
{
\footnotesize
\begin{tabular}{p{0.077\linewidth}C{0.14\linewidth}C{0.14\linewidth}C{0.15\linewidth}C{0.13\linewidth}C{0.19\linewidth}}
\toprule
\textbf{prec.} & \textbf{\hyperref[problem:BPCSC]{$\ProblemBPCSC$}} & \textbf{\hyperref[problem:PCSC]{$\ProblemPCSC$}} & \textbf{\hyperref[problem:PCMSSC]{$\ProblemPCMSSC$}} & \textbf{\hyperref[problem:PCWCDT]{$\ProblemPCWCDT$}} & \textbf{\hyperref[problem:PCACDT]{$\ProblemPCACDT$}} \\
\midrule
\multirow{2}{\linewidth}{none} & $\br{1, \frac{e}{e-1}}$ & $\br{1, \frac{e}{e-1}}$  & $\br{\frac{4e}{e-1},1}$ & $\cO(\log n)$  & $\cO(\log n/\log\log n)$  \\
      &  \cite{TheBudgetedMaximumCoverageProblem} &  \cite{TheBudgetedMaximumCoverageProblem} & \cite{EfficientSequencesOfTrials} &  \cite{DiagnosisDetermination} & \cite{TightAnalysisGreedyUniformDecisionTree} \\

\addlinespace[2pt]
\cdashline{2-6}
\addlinespace[2pt]

inforest & $\br{1, \frac{e}{e-1}}$ Thm.~\ref{thm:BPCSC-inforest} & $\br{1, \frac{e}{e-1}}$ Thm.~\ref{thm:BPCSC-inforest} & $\br{3421.7, \frac{e}{e-1}}$ Cor.~\ref{cor:BPCSC-inforest-to-PCMSSC} & $\cO(\log n)$ Cor.~\ref{cor:PCWCDT-special-cases} & $\cO(\log n)$ \newline Cor.~\ref{cor:PCACDT-special-cases} \\

\addlinespace[2pt]
\cdashline{2-6}
\addlinespace[2pt]

outforest & $\br{\cO(\log n), 4}$ Thm.~\ref{thm:BPCSC-outforest} & $\br{\cO(\log n), 4}$ Thm.~\ref{thm:BPCSC-outforest} & $\br{\cO(\log n), 4}$ Thm.~\ref{thm:BPCSC-outforest} & $\cO(\log^2 n)$ Cor.~\ref{cor:PCWCDT-special-cases} & $\cO(\log^2n)$ \newline Cor.~\ref{cor:PCACDT-special-cases} \\

\addlinespace[2pt]
\cdashline{2-6}
\addlinespace[2pt]

general & $\br{\sqrt{m \cdot H_n}+1, 1}$ Thm.~\ref{thm:BPCSC} & $\br{\cO(\sqrt{m}/f), 2}$ Thm.~\ref{thm:MDPCStoPCSC}, $\br{\sqrt{m \cdot H_n}+1, 1}$ Cor.~\ref{cor:BPCSC-to-PCSC} & $(\sqrt{m \cdot H_n}+1, 1)$ Thm.~\ref{thm:BPCSC-to-PCMSSC} & $\cO(\sqrt{m}\cdot\log n)$ Cor.~\ref{cor:PCWCDT-special-cases} & $\cO(\sqrt{m}\cdot\log^{3/2}n)$ Cor.~\ref{cor:PCACDT-special-cases} \\
\bottomrule
\end{tabular}
}
\end{table}

For arbitrary precedence constraints, we obtain an $\cO(\sqrt{m}\cdot\log n)$-approximation algorithm for \textsc{Precedence Constrained Worst Case Decision Tree} ($\hyperref[problem:PCWCDT]{\ProblemPCWCDT}$) and $\cO(\sqrt{m}\cdot\log^{3/2}n)$-approximation for \textsc{Precedence Constrained Average Case Decision Tree} ($\hyperref[problem:PCACDT]{\ProblemPCACDT}$), where $m$ is the number of tests and $n$ is the number of hypotheses. These results are obtained through reductions to fractional set cover problems ($\hyperref[problem:PCSC]{\ProblemPCSC}$, $\hyperref[problem:PCMSSC]{\ProblemPCMSSC}$). The intuition behind the reduction is that we wish find a sequence of tests (respecting the precedence constraints) whose execution separates the hypotheses into responses of significantly smaller size, and then apply the same procedure recursively on the resulting subproblems. To find such a sequence we construct an appropriate precedence constrained set covering instance, where the sets correspond to tests and the elements correspond to hypotheses. If done correctly, it can be shown that the cost of an optimal solution to the set cover instance (for either of the two cost objectives) is a lower bound on the cost of the decision tree instance, and any approximation for the set cover instance can be converted into an approximation for the decision tree with an additional logarithmic factor loss.

The remaining task is to obtain quality approximations for the set covering problems with precedence constraints. We remark that there already exists an $\cO(\sqrt{m})$-approximation for the \textsc{Precedence-Constrained Min-Sum Set Cover}~\cite{PCMSSC}, but it works only for the case when $f=1$. As it turns out, the case $f\neq 1$ is significantly more challenging. Moreover, this generalization is required for our decision tree algorithms to work (namely we require $f=1/4$). To remedy this issue, we develop new $\bigo\br{\sqrt{m\cdot\log n}}$-approximation algorithms for both cost criteria. The key in our analysis is to analyze yet another variant of set covering, which we call \textsc{Budgeted Precedence Constrained Set Cover} (denoted as $\ProblemBPCSC)$. The difference is that now we wish to find a (precedence-closed) subset of sets of size at most $b$ that covers as many elements as possible. Again, it turns out that one can reduce both of the previously mentioned covering problems to the budgeted version. For the $\ProblemPCSC$, it suffices to guess the value of $b$ to be equal to the cost of the optimal solution, and then apply the approximation for $\ProblemBPCSC$. However, for the $\ProblemPCMSSC$, such idea does not work. Instead, we find a good covering sequence by iteratively allowing exponentially growing budgets and applying the approximation for $\ProblemBPCSC$ at each step. Intuitively, since we start from a very low budget, this allows the algorithm to choose which elements need to be covered quickly. Our approach is reminiscent of the ideas used for solving problems regarding \textsc{Minimum Latency TSP}~\cite{PathsTreesMinimumLatencyTours,KTravelingRepairmenProblem} and \textsc{Adaptive TSP}~\cite{ApproxAlgsForOptDTsAndAdapTSPProblems}. The overarching theme of this connection is that one can treat appending sets to the covering sequence as a kind of tour on the DAG induced by the precedence constraints, and the cost of covering an element as a latency of ''visiting'' that element. Note that however, this intuition is not exact and is never used directly neither in the procedure nor its analysis.

Next, to obtain an $\br{\bigo\br{\sqrt{m\cdot\log n}}, 1}$-bicriteria approximation for $\ProblemBPCSC$, (we may exceed to budget $\bigo\br{\sqrt{m\cdot\log n}}$ times) we employ an idea inspired by the previously mentioned algorithm of \cite{PCMSSC}, however our analysis is more intricate. Here, the algorithm follows a greedy scheme similar to the one used for the classical \textsc{Set Cover}, but instead of selecting a single set at each step, it picks a precedence-closed family of sets with a good  (number of covered elements per set). Note, that for our needs we additionally need to control the size of the selected family. We call this problem \textsc{Budgeted Max-Density Precedence-Closed Subfamily} ($\ProblemBMDPCS$). We show that an approximation for $\ProblemBMDPCS$ can be used to obtain an approximation for $\ProblemBPCSC$. To solve $\ProblemBMDPCS$ we then use yet another greedy scheme, by simply choosing the set with its closure (if smaller then $b$) which covers the most new elements per unit of size. This algorithm has an $b$-approximation which may be as large as $\Omega\br{n}$, however by a clever analysis can be used to obtain an $\br{\sqrt{m\cdot H_n}+1, 1}$-approximation for $\ProblemBPCSC$. Using similar ideas we also develop an algorithm for $\ProblemPCSC$, whose approximation ratio depends on the required fraction of covered elements $f$ and is $\cO\br{\sqrt{m}/f}$. Since in our applications $f$ is constant this gives a slightly better approximation ratio for $\ProblemPCWCDT$.

At last we also concern ourselves with some special cases of precedence constraints. For the inforests, where each test/set has at most one child, we show that $\ProblemBPCSC$ can be encoded as an instance of the well-known \textsc{Maximum Coverage}. This allows us to obtain an $\cO(\log n)$-approximation for $\ProblemPCWCDT$ and $\ProblemPCACDT$. For the case of outforests, for which every test/set has at most one parent, we show that $\ProblemBPCSC$ can be reduced the \textsc{Group Steiner Orienteering Problem} on trees, which yields an $\br{\cO\br{\log n}, 1}$-bicriteria approximation for $\ProblemBPCSC$ and as a consequence an $\cO\br{\log^2 n}$-approximation for $\ProblemPCWCDT$ and $\ProblemPCACDT$. For this subcase, the connection between our problems and graph exploration becomes apparent, since any sequence of sets becomes a tour of the precedence graph.

\paragraph{Hardness Results.} Our hardness results establish matching or nearly matching lower bounds for most variants through a hierarchy of reductions illustrated in Figure~\ref{fig:hardness_diagram}. We systematically reduce from known hard problems such as \textsc{Group Steiner Tree} and \textsc{Planted Dense Subgraph Conjecture} to precedence constrained covering problems and then to precedence constrained decision tree problems, showing that achieving better approximation ratios is computationally hard even for special cases. Our hardness results, including state-of-the-art, are summarized in Table~\ref{tab:hardness}.

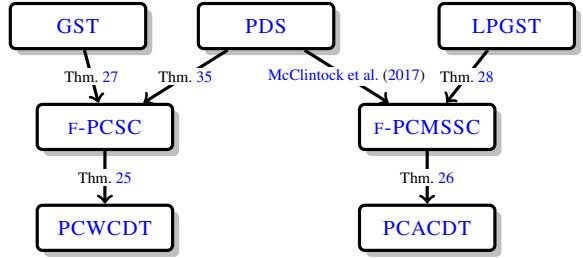
\begin{figure}[ht!]
    \centering
    \begin{tikzpicture}[
    node distance=0.7cm and 2cm,
    problem/.style={rectangle, draw, very thick, minimum width=1.8cm, minimum height=0.6cm, align=center, font=\scriptsize, fill=white, drop shadow, rounded corners=2pt},
    arrow/.style={->, very thick},
    edgelabel/.style={font=\tiny, fill=white, inner sep=1pt}
]

% Define positions based on the hand-drawn diagram
% Top level
\node[problem] (gst) at (-2, 0) {$\hyperref[problem:GST]{\ProblemGST}$};
\node[problem] (pds) at (0.5, 0) {$\hyperref[problem:PDS]{\ProblemPDS}$};
\node[problem] (lpgst) at (3.7, 0) {$\hyperref[problem:LPGST]{\ProblemLPGST}$};

% Middle level
\node[problem, below left=0.7cm and 0.3cm of pds] (fpcsc) {$\hyperref[problem:PCSC]{\ProblemPCSC}$};
\node[problem, below right=0.7cm and 0.3cm of pds] (fpcmssc) {$\hyperref[problem:PCMSSC]{\ProblemPCMSSC}$};

% Bottom level
\node[problem, below=0.7cm of fpcsc] (pcwcdt) {$\hyperref[problem:PCWCDT]{\ProblemPCWCDT}$};
\node[problem, below=0.7cm of fpcmssc] (pcacdt) {$\hyperref[problem:PCACDT]{\ProblemPCACDT}$};

% Arrows with theorem references
\draw[arrow] (gst) -- node[edgelabel, pos=0.5] {Thm.~\ref{thm:PCSC-outforest-hardness}} (fpcsc);
\draw[arrow] (lpgst) -- node[edgelabel, pos=0.5] {Thm.~\ref{thm:PCMSSC-outforest-hardness}} (fpcmssc);
\draw[arrow] (pds) -- node[edgelabel, pos=0.5] {Thm.~\ref{thm:PCSC-PDS-hardness}} (fpcsc);
\draw[arrow] (pds) -- node[edgelabel, pos=0.5] {\cite{PCMSSC}} (fpcmssc);
\draw[arrow] (fpcsc) -- node[edgelabel, pos=0.5] {Thm.~\ref{thm:PCSC-to-PCWCDT-hardness}} (pcwcdt);
\draw[arrow] (fpcmssc) -- node[edgelabel, pos=0.5] {Thm.~\ref{thm:PCMSSC-to-PCACDT-hardness}} (pcacdt);

\end{tikzpicture}
    \caption{Inapproximability relations between problems, $\Pi_1 \to \Pi_2$ denotes that an inapproximability result for problem $\Pi_1$ implies an inapproximability result for problem $\Pi_2$.}
    \label{fig:hardness_diagram}
\end{figure}

\begin{table}[ht!]
\caption{Inapproximability results. PDS $\equiv$ claims conditional on the Planted Dense Subgraph Conj.}
\label{tab:hardness}
\centering
{
\footnotesize
\begin{tabular}{p{0.08\linewidth}C{0.19\linewidth}C{0.19\linewidth}C{0.19\linewidth}C{0.19\linewidth}}
\toprule
\textbf{precedenses} & \textbf{\hyperref[problem:PCSC]{\ProblemPCSC}} & \textbf{\hyperref[problem:PCMSSC]{\ProblemPCMSSC}} & \textbf{\hyperref[problem:PCWCDT]{\ProblemPCWCDT}} & \textbf{\hyperref[problem:PCACDT]{\ProblemPCACDT}} \\
\midrule
\multirow{2}{\linewidth}{none} & $(1-o(1))\cdot H_n$ & $(4-\epsilon)$ & $o(\log n)$ & $\br{4-\epsilon}$ \\
& \cite{AnalyticalApproachToParallelRepetition} & \cite{ApproximatingMinSumSetCover} & \cite{HardnessOfMinHeightDTP} & \cite{DTsforEntIdent} \\

\addlinespace[2pt]
\cdashline{2-5}
\addlinespace[2pt]

\multirow{2}{\linewidth}{outforest} & $o(\log^{2} n)$ & $o(\log^{2} n)$ & $o(\log^{2} n)$ & $(4-\epsilon)$ \\
& Thm.~\ref{thm:PCSC-outforest-hardness} & Thm.~\ref{thm:PCMSSC-outforest-hardness} & Cor.~\ref{thm:PCWCDT-outforest-hardness} & \cite{DTsforEntIdent} \\

\addlinespace[2pt]
\cdashline{2-5}
\addlinespace[2pt]

\multirow{2}{\linewidth}{general} & $o(n^{1/6})$, $o(m^{1/12})$ & $o(n^{1/6})$, $o(m^{1/12})$ & $o(n^{1/6})$, $o(m^{1/12})$ & $o(n^{1/6})$, $o(m^{1/12})$ \\
& Thm.~\ref{thm:PCSC-PDS-hardness}, PDS & \cite{PCMSSC}, PDS & Cor.~\ref{thm:PCWCDT-PDS-hardness}, PDS & Cor.~\ref{thm:PCACDT-PDS-hardness}, PDS \\
\bottomrule
\end{tabular}
}
\end{table}

For arbitrary precedence constraints, we show that under the \textsc{Planted Dense Subgraph Conjecture} ($\ProblemPDS$), no algorithm can achieve $o(m^{1/12})$ nor $o(n^{1/6})$-approximation for $\ProblemPCSC$. We do this by adjusting the reduction from $\ProblemPDS$ to $\ProblemPCMSSC$ of \cite{PCMSSC}. We also show that $\ProblemPCWCDT$ and $\ProblemPCACDT$ generalize $\ProblemPCSC$ with $f=1$ and $\ProblemPCMSSC$ with $f=1$ respectively, which carries on the inapproximability results for the decision tree problems as well. 

For the special case of outforests, for any $\epsilon>0$ we establish $\cO(\log^{2-\epsilon} n)$ hardness of approximation for $\ProblemPCSC$ by reducing from \textsc{Group Steiner Tree} on trees, which is known to be hard to approximate within $\cO(\log^{2-\epsilon} n)$ factor, unless $\text{NP}\subseteq \text{ZTIME}(n^{\text{polylog}(n)})$. This hardness result naturally extends to $\ProblemPCWCDT$ via our reduction hierarchy. A similar result also holds for $\ProblemPCMSSC$, however here $f\neq 1$ and thus the reduction does not carry to $\ProblemPCACDT$.

\subsection{Organization}
Section~\ref{sec:preliminaries} introduces the necessary notions and preliminaries including all problem definitions. Section \ref{sec:AL} presents our approximation algorithms for decision trees with precedence constraints, relying on approximation algorithms for set covering with precedence constraints developed in Section~\ref{sec:SC}.
The latter use algorithms for density-type problems in Section~\ref{sec:MDPCS}.
Section~\ref{sec:hardness} provides approximation hardness results for both kinds of the aforementioned problems as well as hardness of binary search with precedence constraints.

\section{Preliminaries and problem definitions} \label{sec:preliminaries}

%For any set $X$, we denote by $\spr{X}$ its cardinality.
%For any integer $k \geq 1$, we use $[k]$ to denote the set $\brc{1, 2, \ldots, k}$.

For any integer $k \geq 1$, we use $[k]$ to denote the set $\brc{1,\ldots,k}$. 
Let $\sigma$ be a sequence and let $\ell\in[\spr{\sigma}]$. By $\pref{\sigma}{\ell}$ we mean the prefix of $\sigma$ of length $\ell$. For a given rooted tree $D=\br{V\br{D}, E\br{D}}$ we denote its root by $r\br{D}$. 
For any node $v$ of $D$, $D_v$ is the subtree of $D$ that consists of $v$ and all its descendants.

% \paraTitle{Decision tree problems.}
% In the decision tree problems, we are given a set of hypotheses $\cH$ and a set of tests $\cT$. Each test $t$ is an arbitrary partition of $\cH$ into disjoint subsets $U_{t,1}, U_{t,2}, \ldots, U_{t,r_t}$, where $r_t$ is the number of possible responses to test $t$. When a test $t$ is performed, the response indicates which subset $U_{t,j}$ contains the hidden hypothesis $\target \in \cH$. The tests are subject to precedence constraints encoded by a DAG $\cF = \brc{\cT, \preceq}$, meaning that a test $t$ can be performed only if all its predecessors in $\cF$ have already been performed.

\paraTitle{Partial orders and precedence constraints.}
Throughout this paper, we work with POSETs $(\cX, \preceq)$, where $\preceq$ is a reflexive, transitive, and antisymmetric binary relation.
For elements $X, Y \in \cX$, we say that $Y$ is a \emph{predecessor} of $X$ if $Y \preceq X$, and that $X$ is a \emph{successor} of $Y$.
A subset $\cS \subseteq \cX$ is \emph{precedence-closed} if for all $X \in \cS$ and all $Y \in \cX$ such that $Y \preceq X$, it holds that $Y \in \cS$.
A sequence $(X_1, X_2, \ldots, X_k)$ of elements from $\cX$ is \emph{precedence-closed} if for any $X_i$ and $X_j$ such that $X_i \preceq X_j$, we have $i < j$. We consider two special types of partial orders: An \emph{inforest} is a partial order in which every element has at most one successor. An \emph{outforest} is a partial order in which every element has at most one predecessor.

\paraTitle{Decision trees.}
In decision tree problems, the goal is to design a \textit{strategy} $\cS$ understood as an adaptive algorithm, that based on the previous replies, gives the \questioner the next test to be performed.
This strategy is typically represented as a \textit{decision tree}, which is a rooted tree $D$ where internal nodes represent tests to be performed, edges encode possible replies and leaves represent the hypotheses.
Formally, a decision tree is defined recursively as follows.
If $r=\brc{H_1,\ldots,H_l}\in\cT$ is the first test performed by $\cS$, then $r$ is the root of $D$.
For each reply $H_i$, $i\in[l]$, take inductively defined decision tree $D_i$ derived from the tests done by $\cS$ when the reply to $r$ is $H_i$.
Then, $D$ is obtained by making the roots of $D_1,\ldots,D_l$ the children of $r$, and the edge connecting $r$ to the root of $D_i$ is labeled with the reply $H_i$.
To complete the inductive definition, if $\cS$ is a trivial strategy that outputs the target hypothesis $h^*\in \cH$, then the corresponding tree $D$ is a leaf labeled with $h^*$.
It should be remarked that it is possible for a test to appear multiple times in the decision tree.
However, with a slight abuse of the notation we assume that non-leaf vertices are a subset of $\cT$ whenever there is no ambiguity. Similarly, we may assume that the set of leaves of $D$ is $\cH$. When the tests are subject to precedence constraints $(\cT, \preceq)$, a decision tree $D$ is called \emph{precedence-closed} if for every node $v$ in $D$, the sequence of tests on the path from $r\br{D}$ to $v$ is precedence-closed. This means that a test $t$ is executed only if all its predecessors in $(\cT, \preceq)$ have already been performed.

\paraTitle{Cost measures for decision trees.}
Let $\tests{D}{v}$ denote the sequence of tests on the path from $r\br{D}$ to $v$ including $v$ if it denotes a test. The cost of identifying a hypothesis $h$ using a decision tree $D$ is $\COST\br{D, h}=\spr{\tests{D}{h}}$.
We consider two cost measures for decision trees: the worst-case cost 
\[
\COSTW\br{D} = \max_{h \in \cH} \brc{\COST\br{D, h}}
\] 
and the average-case cost\footnote{Here, for simplicity we do not divide by the number of hypotheses $n$, since this does not affect the approximation ratios.} 
\[
\COSTA\br{D} = \sum_{h \in \cH} \COST\br{D, h}.
\]
%The criteria we consider for any $\ProblemPCDT$ instance $\cI$ are denoted by
%$$\OPTA\br{\cI}=\min\brc{\COSTA\br{D} \mid D\textup{ is a decision tree}},$$
%$$\OPTW\br{\cI}=\min\brc{\COSTW\br{D} \mid D\textup{ is a decision tree}}.$$
%\DD{For $X\subseteq\cH$, $p(X)=\sum_{x\in X}p(x)$.}
%(Whenever the criterion is not important or we want to make a claim that applies both to $\ProblemPCWCDT$ and $\ProblemPCACDT$, we use the symbol $\ProblemPCDT$ to refer to a \emph{Precedence Constrained Decision Tree} instance).

\medskip
Having introduced the necessary notation, we now formally state the decision tree problems for which we provide algorithms in Section~\ref{sec:AL}.
Throughout the paper we use $\cI$ to denote instances of particular problems, and $\OPT\br{\cI}$ denotes then the cost of an optimal solution for $\cI$.
\begin{problem}[\textsc{Precedence Constrained Worst Case Decision Tree} ($\ProblemPCWCDT$)] \label{problem:PCWCDT}
Given a set of hypotheses $\cH$, a set of tests $\cT$ and a precedence relation $(\cT,\preceq)$ on tests, find a precedence-closed decision tree $D$ that minimizes the worst-case cost $\COSTW\br{D}=\OPT\br{\cI}$.
\end{problem}
\begin{problem}[\textsc{Precedence Constrained Average Case Decision Tree} ($\ProblemPCACDT$)] \label{problem:PCACDT}
Given a set of hypotheses $\cH$, a set of tests $\cT$ and a precedence relation $(\cT,\preceq)$ on tests, find a precedence-closed decision tree $D$ that minimizes the average-case cost $\COSTA\br{D}=\OPT\br{\cI}$.
\end{problem}

\paraTitle{Additional notation for decision trees.}
Suppose that a decision tree $D$ performed some tests $T\subseteq\cT$.
The set $\cH_T(D)\subseteq\cH$ consists of all hypotheses $h$ such that for each test $t\in T$ the reply is a set $X\in t$ such that $h\in X$ and $X\subseteq \cH_T(D)$.
Note that $\cH_T(D)$ does not depend on the order of performing the tests in $T$.
Likewise, if $v$ is any node of a decision tree $D$ then $\cH_v(D)=\cH_T(D)$, where $T=\tests{D}{v}$.
We usually write $\cH_T$ when $D$ is clear from the context.

\medskip
\paraTitle{Set covering problems.}
%For any partial order $(X,\preceq)$ and $S\subseteq X$, we say that $S$ is \emph{precedence-closed} if for all $x \in S$ and all $y \in X$ such that $y \preceq x$, it holds that $y \in S$.
For a family of sets $(\cS,\preceq)$ and a set $S\in\cS$, let $\closure{S}$ be the minimal precedence-closed subset of $\cS$ containing $S$.
For a family of sets $\cC\subseteq \cS$ its \emph{coverage} is defined as $\cov\br{\cC}=\bigcup_{C\in\cC}C$, and by extension, $\cov\br{\cC, X} = \cov\br{\cC} \cap X$ for any set $X$.
If $\cC$ is a sequence of sets, rather than a family, we analogously define $\cov\br{\cC}=\bigcup_{C\in\cC}C$.
Let $\cC=\angl{C_1,\ldots,C_l}$. We say that $\cC$ is \emph{precedence-closed} if for any $C_i$ and $C_j$ such that $C_i\preceq C_j$ it holds $i<j$.
The \emph{coverage time} of an element $x$, $\coverageTime\br{x,\cC}$, in such precedence-closed sequence $\cC$ is the minimum index $i$ such that $x\in C_i$ if such index exists, i.e., if $x\in\cov\br{\cC}$, and $\spr{\cC}$ whenever $x\notin \cov\br{\cC}$.
Given a universe $\cU$, we denote by $\coverageTime\br{\cC, \cU}=\sum_{x\in\cU}\coverageTime\br{x,\cC}$ the sum of coverage times of all elements in $\cU$.

\begin{problem}[\textsc{Precedence Constrained Set Cover} ($\ProblemPCSC$)] \label{problem:PCSC}
Given a universe $\cU$ of $n$ items, a family $\cS$ of $m$ subsets of $\cU$, a precedence relation $(\cS,\preceq)$ and $0<f\leq 1$, find a precedence-closed subfamily $\cC\subseteq\cS$ that $\spr{\cov\br{\cC,\cU}}\geq f\cdot n$ items and minimizes $\spr{\cC}$.
\end{problem}
\begin{problem}[\textsc{Precedence Constrained Min-Sum Set Cover} ($\ProblemPCMSSC$)] \label{problem:PCMSSC}
Given a universe $\cU$ of $n$ items, a family $\cS$ of $m$ subsets of $\cU$, a precedence relation $(\cS,\preceq)$ and $0<f\leq 1$, find a precedence-closed sequence $\cC$ that $\spr{\cov\br{\cC,\cU}}\geq f\cdot n$ items and minimizes $\coverageTime\br{\cC, \cU}$.
\end{problem}
\begin{problem}[\textsc{Budgeted Precedence Constrained Set Cover} ($\ProblemBPCSC$)] \label{problem:BPCSC}
Given a universe $\cU$ of items, a family $\cS$ of subsets of $\cU$, a precedence relation $(\cS,\preceq)$ and a budget $\budget>0$, find a precedence-closed set cover $\cC$ that maximises $\spr{\cov\br{\cC,\cU}}$ and satisfies $\spr{\cC}\leq \budget$.
\end{problem}

\medskip
\paraTitle{Maximum density problems.}
The algorithms for set covering are obtained via reductions from a problem of finding a maximum density subset (cf. Section~\ref{sec:MDPCS}).
The \emph{density} of a nonempty family $\cC$ on a set $X$ is
$
\Delta\br{\cC, X} = \frac{\spr{\cov\br{\cC, X}}}{\spr{\cC}}.
$
We write $\Delta\br{\cC}$ when $X$ is the universe.

\begin{problem}[\textsc{Max-Density Precedence-Closed Subfamily} ($\ProblemMDPCS$)] \label{problem:MDPCS}
Given a universe $\cU$ of $n$ items, a family $\cS$ of $m$ subsets of $\cU$ and a precedence relation $(\cS,\preceq)$, find a precedence-closed subfamily $\cC \subseteq \cS$ that maximizes $\Delta\br{\cC}$.
\end{problem}

\begin{problem}[\textsc{Bounded Max-Density Precedence-Closed Subfamily} ($\ProblemBMDPCS$)] \label{problem:BMDPCS}
Given a universe $\cU$ of $n$ items, a family $\cS$ of $m$ subsets of $\cU$, a precedence relation $(\cS,\preceq)$ and an integer $\budget>0$, find a precedence-closed subfamily $\cC \subseteq \cS$ that maximizes $\Delta\br{\cC}$ and $\spr{\cC}\leq \budget$.
\end{problem}

\medskip
\paraTitle{Steiner problems.}
Let $(V,d)$ denote a symmetric metric with distance function $d: V \times V \to \mathbb{R}_{\geq 0}$ satisfying the triangle inequality.
An \emph{$r$-tour} $\tau$ is a walk that starts and ends at $r \in V$. Denote by $\spr{\tau}$ the length of tour $\tau$ with respect to $d$.
Let $\cX$ be a family of \emph{groups}, where for each $X \in \cX$, $X\subseteq V$.
Let $w: \cX \to \mathbb{R}_{\geq 0}$ be a weight assignment to groups.
Let $\cov\br{\tau} \subseteq \cX$ denote the set of all groups \emph{covered} by tour $\tau$, that is, $\cov\br{\tau} = \brc{X \in \cX\colon X \cap \tau \neq \emptyset}$ where we slightly abuse notation and write $X \cap \tau$ for the intersection of $X$ with the set of vertices visited by $\tau$.
% For a tour $\tau$ and a group $X \in \cX$, define $t_\tau(X)$ as the length of the shortest prefix of $\tau$ (with respect to $d$) that visits at least one vertex from $X$; if $X \notin \cov\br{\tau}$, we set $t_\tau(X) = \spr{\tau}$.

\begin{problem}[\textsc{Group Steiner Tree} (GST)] \label{problem:GST}
Given a metric space $(V, d)$, a root vertex $r \in V$ and a family of groups $\cX \subseteq 2^V$, find an $r$-tour such that $\cov\br{\tau} = \cX$ minimizing $\spr{\tau}$.
\end{problem}

\begin{problem}[\textsc{Group Steiner Orienteering} (GSO)] \label{problem:GSO}
Given a metric space $(V, d)$, a root vertex $r \in V$, a family of groups $\cX \subseteq 2^V$, a weight function $w: \cX \to \mathbb{R}_{\geq 0}$, and a budget $\budget > 0$, find an $r$-tour $\tau$ with $\spr{\tau} \leq \budget$ that maximizes
$\sum_{X \in \cov\br{\tau}} w(X)$.
\end{problem}
% \begin{problem}[\textsc{Partial Latency Group Steiner Tree} (LPGST)] \label{problem:LPGST}
% Given a metric space $(V, d)$, a root vertex $r \in V$, a family of groups $\cX \subseteq 2^V$, a weight function $w: \cX \to \mathbb{R}_{\geq 0}$, and a target $h \leq \spr{\cX}$, find an $r$-tour $\tau$ such that $\spr{\cov\br{\tau}} \geq h$ and that minimizes
% $c(\tau) = \sum_{X\in \cX} w(X) \cdot t_\tau(X)$.
% %$$c(\tau) = \sum_{X \in \cov\br{\tau}} w(X) \cdot t_\tau(X) + \sum_{X \in \cX \setminus \cov\br{\tau}} w(X) \cdot \spr{\tau}.$$
% \end{problem}

\medskip
\paraTitle{Approximation conventions.}
Throughout this paper, when we use bicriteria ratio $(\alpha,\beta)$, which means that an algorithm is allowed to lose a factor $\beta$ on the coverage requirement (equivalently, it covers at least a $1/\beta$-fraction of the target amount), while simultaneously losing at most a factor $\alpha$ on the cost side, i.e., it either pays at most $\alpha$ times the optimum cost or violates the given budget by at most a factor $\alpha$, depending on the problem formulation.

\section{Decision Trees via Covering Problems} \label{sec:AL}

In this section we present approximation algorithms for both $\ProblemPCWCDT$ and $\ProblemPCACDT$ by reducing them to $\ProblemPCSC$ and $\ProblemPCMSSC$, respectively. The overarching idea of these reductions is to use the solution of the appropriate covering problem in order to find a precedence-closed sequence of tests, such that performing all of them, guarantees that the size of each response is significantly smaller than $\spr{\cH}$. Moreover, the cost of performing these tests should be $\cO\br{\OPT}$ (up to the approximation ratio of the covering problem). The rest of the decision tree is then built recursively and since the depth of the recursion is $\cO\br{\log n}$ and as a consequence the loss in the approximation ratio is $\cO\br{\log n}$ as well. In what follows we make these ideas precise. 

We will require some additional notation.
%For any node $v$ of the $D$ by $\cH_v$ we will mean the set of hypotheses such that the test $t$ associated with $v$ is performed in $D$ if the set of possible hypotheses is $\cH_v$.
Let $\cH'\subseteq \cH$. By $t\br{\cH'}$ we denote the test $t$ restricted to hypotheses in $\cH'$, i.e., for each reply $H\in t$, we have a reply $H\cap \cH'$ in $t\br{\cH'}$ (if it is nonempty). By $D\setminus V$ we denote the collection of subtrees (a forest) obtained by the removal of all nodes in $V$ from $D$.
We will use the following immediate monotonicity properties. 
\begin{lemma}\label{lemma:subspace_opt}
    Let $\cI=\br{\cH, \cT,\preceq}$ and $\cI'=\br{\cH', \cT, \preceq}$ be two instances of $\ProblemPCACDT$ or $\ProblemPCWCDT$, where $\cH'\subseteq \cH$.
    Then, $\OPT\br{\cI'} \leq \OPT\br{\cI}$.
\end{lemma}
\begin{lemma}\label{lemma:subinstances_avg}
    Let $\cI=(\cH, \cT, \preceq)$ be an instance of $\ProblemPCACDT$.
    Let $\cH_1,\dots,\cH_k\subseteq \cH$ be such that $\bigcup_{i=1}^k \cH_i \subseteq \cH$ and for any $i\neq j$, $\cH_i\cap \cH_j = \emptyset$.
    Then, 
    \[
    \OPT\br{\cI} \geq \sum_{i=1}^k \OPT\br{\br{\cH_i, \cT, \preceq}}
    \]
\end{lemma}

The following definition gives a basic tool for analysis of our algorithms.
\begin{definition}[Sepcover]
  Let $D$ be any decision tree for $\cI=\br{\cH, \cT, \cF}$.
  Let $v$ be a node that is closest to the root and has the property that after removal of $\tests{D}{v}$ from $D$, each subtree has at most $\spr{\cH}/2$ leaves.
  Then the set $\tests{D}{v}$ is called a \emph{sepcover} for $D$ and is denoted by $\sepcoverD{D}$.
  If $\COSTW\br{D}=\OPT\br{\cI}$, then we write $\sepcover{\cI} = \sepcoverD{D}$ (ties broken arbitrarily).
\end{definition}
Figure \ref{fig:sepcover} illustrates the definition of sepcover.

\begin{figure}[h]
\centering
\begin{tikzpicture}[scale=1.3]
    % Test nodes in sequence (nodes with labels)
    \node[circle, fill=black, inner sep=5pt, label=center:{}] (t1) at (0, 0) {};
    \node[circle, fill=black, inner sep=5pt, label=center:{}] (t2) at (2, -0.7) {};
    \node[circle, fill=black, inner sep=5pt, label=center:{}] (t3) at (4, -1.4) {};
    \node[circle, fill=black, inner sep=5pt, label=center:{}] (t4) at (7, -2.45) {};
    
    % Connecting lines in the main sequence with dots
    \draw[very thick] (t1) -- (t2);
    \draw[very thick] (t2) -- (t3);
    \draw[very thick] (t3) -- (5.25, -1.84);
    \node[rotate=-20] at (5.5, -1.93) {$\dots$};
    \draw[very thick] (5.75, -2.02) -- (t4);
    
    % Subtrees from first test (3 triangles, decreasing size)
    \draw[very thick] (t1) -- (-1.1, -0.65);
    \draw[very thick, fill=gray!20, drop shadow] 
        (-1.1,-0.65) -- (-1.55,-1.45) -- (-0.65,-1.45) -- cycle;
    
    \draw[very thick] (t1) -- (-0.1, -0.65);
    \draw[very thick, fill=gray!20, drop shadow] 
        (-0.1,-0.65) -- (-0.5,-1.35) -- (0.3,-1.35) -- cycle;
    
    \draw[very thick] (t1) -- (0.8, -0.65);
    \draw[very thick, fill=gray!20, drop shadow] 
        (0.8,-0.65) -- (0.45,-1.25) -- (1.15,-1.25) -- cycle;
    
    % Subtrees from second test (4 triangles, decreasing size)
    \draw[very thick] (t2) -- (1.1, -1.35);
    \draw[very thick, fill=gray!20, drop shadow] 
        (1.1,-1.35) -- (0.65,-2.15) -- (1.55,-2.15) -- cycle;
    
    \draw[very thick] (t2) -- (2.0, -1.35);
    \draw[very thick, fill=gray!20, drop shadow] 
        (2.0,-1.35) -- (1.6,-2.1) -- (2.4,-2.1) -- cycle;
    
    \node at (2.6, -1.6) {$\dots$};
    
    \draw[very thick] (t2) -- (3.2, -1.35);
    \draw[very thick, fill=gray!20, drop shadow] 
        (3.2,-1.35) -- (2.85,-2.0) -- (3.55,-2.0) -- cycle;
    
    % Subtrees from third test (2 triangles, decreasing size)
    \draw[very thick] (t3) -- (3.8, -1.95);
    \draw[very thick, fill=gray!20, drop shadow] 
        (3.8,-1.95) -- (3.4,-2.75) -- (4.2,-2.75) -- cycle;
    
    \draw[very thick] (t3) -- (4.9, -1.95);
    \draw[very thick, fill=gray!20, drop shadow] 
        (4.9,-1.95) -- (4.55,-2.65) -- (5.25,-2.65) -- cycle;
    
    % Subtrees from fourth test (3 triangles, smallest)
    \draw[very thick] (t4) -- (5.9, -2.95);
    \draw[very thick, fill=gray!20, drop shadow] 
        (5.9,-2.95) -- (5.5,-3.65) -- (6.3,-3.65) -- cycle;
    
    \draw[very thick] (t4) -- (6.9, -2.95);
    \draw[very thick, fill=gray!20, drop shadow] 
        (6.9,-2.95) -- (6.55,-3.6) -- (7.25,-3.6) -- cycle;
    
    \draw[very thick] (t4) -- (7.8, -2.95);
    \draw[very thick, fill=gray!20, drop shadow] 
        (7.8,-2.95) -- (7.5,-3.55) -- (8.1,-3.55) -- cycle;
    
    % Diagonal brace along the sequence
    \draw [very thick, decorate, 
       decoration={brace, amplitude=20pt, mirror}] 
      (7.35, -2.5) -- (-0.2, 0.25);
    
    % Label for the brace
    \node at (3.85, -0.45) {$P_D$};
    
\end{tikzpicture}
\caption{Sepcover sequence in a decision tree}
\label{fig:sepcover}
\end{figure}
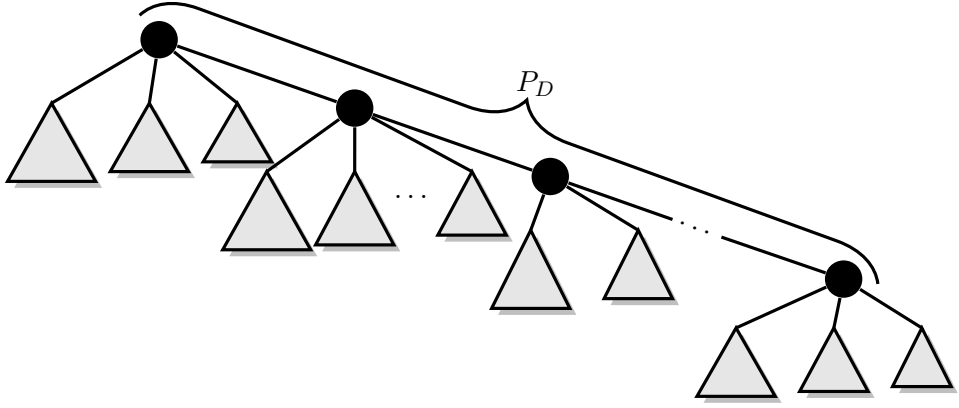

Intuitively, our idea is to construct a set covering instance such that the optimal solution to this instance is connected to the cost of tests in the sepcover of an optimal decision tree for the input instance $\cI$. This is done as follows.
%We will say that a test in $\sepcover{\cI}$ \emph{sepcovers} an element $u \in \cH$ in $C$ if after applying the test $t$ in $C$, $u$ belongs to a reply $\cH'$ of size at most $\spr{\cH}/2$.
Consider a set of hypotheses $\cH$ and the corresponding tests $\cT$.
For each test $t\in\cT$, let
\[
\xi(t)=\brc{h\in\cH \mid h\in H \textup{ for some }H\in t\textup{ s.t. }  \spr{H}\leq\frac{3}{4}\cdot\spr{\cH}}.
\]
Then, $\xi(\cT)=\brc{\xi(t)\mid t\in\cT}$ and if $\preceq$ is a partial order on $\cT$, then $\preceq_{\xi}$ is the corresponding partial order on $\xi(\cT)$, i.e., $t\preceq t'$ if and only if $\xi(t)\preceq_{\xi}\xi(t')$.
By extension, for any subset $S\subseteq\cT$, $\xi(S)=\brc{\xi(t) \mid t\in S}$.
We also set $\cU=\bigcup_{t\in\cT}\xi(t)$.
\begin{lemma}
    We have $\spr{\cU}\geq \spr{\cH}-1$ and $\spr{\cU}\geq \frac{4}{5}\cdot\spr{\cH}$.
\end{lemma}
\begin{proof}
    If $\spr{\cH}\leq4$, then there exists a test $t$ which separates some pair of hypotheses and thus for every $H\in t$, $\spr{H}\leq\frac{3}{4}\cdot\spr{\cH}$, and therefore $\xi\br{t}=\cH$. Assume that $\spr{\cH}\geq5$. We show, that there is at most one hypothesis that is not in $\cU$. To see this observe that if a hypothesis $h_1\notin \cU$, then for every test $t\in\cT$, it holds that for the reply $H\in t$ containing $h_1$, $\spr{H}>\frac{3}{4}\cdot\spr{\cH}$. Assume that there is a second hypothesis $h_2\notin \cU$. By definition we know that there exists a test which distinguishes $h_1$ and $h_2$, so $h_2$ needs to belong to a reply of size at most $\frac{3}{4}\cdot\spr{\cH}$, a contradiction. Thus $\spr{\cU}\geq \min_{i\geq5}\brc{\frac{i-1}{i}}\cdot\spr{\cH}=\frac{4}{5}\cdot\spr{\cH}$.
\end{proof}
 
%In our construction, in order to solve a $\ProblemPCDT$ for an instance $\cI=\br{\cH, \cT, \preceq}$, we will apply an algorithm for $\ProblemPCSC$ with an input $\cI=\br{\cH, \xi(\cT), \preceq_{\xi}, \spr{\cH}/4}$.

\medskip
The idea behind the algorithms for the $\ProblemPCWCDT$ (Section~\ref{subsection:PCWCDT}) and the $\ProblemPCACDT$ (Section~\ref{subsection:PCACDT}) is similar.
Hence we extract here the common generic subroutine (see Algorithm~\ref{alg:worstDecisionTree}).
The only difference is a black-box procedure (we call it $\blackBox$) that is used to solve a particular subproblem - either $\ProblemPCSC$ or $\ProblemPCMSSC$ respectively.
Let $\cC$ be the cover obtained for the instance $\br{\cU, \xi\br{\cT}, \preceq_{\xi}, f=1/4}$ by the black-box procedure.
Assume that $\cC$ contains some ordering of its elements. If $\cC$ is the output of an approximation algorithm for $\ProblemPCSC$, then this ordering is a topological ordering of the tests in $\cC$ according to the precedence constraints.
% Note that $\cC$ is supposed to play the role of $\sepcover{\cI}$ according to the lower bounds in Lemmas~\ref{lemma:cover_sep_worst_case} and~\ref{lemma:cover_sep_average_case} given in the subsections below.
We iteratively build a decision tree $D$ by appending tests associated with consecutive sets from $\cC$.
During each such iteration, we look at all possible replies to the appended test and for each sufficiently small reply $\cH'$ ($\spr{\cH'} \leq \frac{5\beta-1}{5\beta}\cdot \spr{\cH}$), we recursively call \ProcDecisionTree on the subinstance induced by $\cH'$. If there exists a large reply $\cH'$ ($\spr{\cH'} > \frac{5\beta-1}{5\beta}\cdot \spr{\cH}$), then we continue appending tests from $\cC$ as a reply associated with $\cH'$. Note that since $\frac{5\beta-1}{5\beta}>\frac{1}{2}$ for any $\beta \geq 1$, there is at most one such reply.
% We also note that the sequence of tests $t_1,t_2\ldots$ in Algorithm~\ref{alg:worstDecisionTree} is a `substitute' of sepcover (cf. Lemma~\ref{lemma:cover_sep_worst_case}).
\begin{algorithm}[h]
\DontPrintSemicolon
\LinesNumbered

\caption{A generic approximation algorithm for $\ProblemPCACDT$ and $\ProblemPCWCDT$.}
\label{alg:worstDecisionTree}
\SetKwProg{Proc}{procedure}{}{}
\Proc{\ProcDecisionTree$(\cH, \cT, \preceq, \blackBox)$}{
\If{$\spr{\cH} = 1$}{
    \Return the trivial decision tree with a single leaf corresponding to the only hypothesis in $\cH$.
}
%\ForEach{$t \in \cT$}{
%Set $t$ to cover $u\in \cH$ if for $u \in U_{t,j}$, $\spr{U_{t,j}}\leq \frac{3}{4}\cdot\spr{\cH}$.
%}
%$S \gets$ Run the $\br{\gamma,\alpha}$-approximation algorithm for $\ProblemPCSC$ on instance $(\cH, \xi(\cT), \preceq_{\xi}, \spr{\cH}/4)$.
$\cC \gets$ Call $\br{\alpha, \beta}$-approximation algorithm $\blackBox$ on instance $(\cU, \xi(\cT), \preceq_{\xi}, f=1/4)$.

$t_{1} \gets$ the test associated with the first set in $\cC$.

$D\gets\br{\brc{t_1}, \emptyset}$.

$\cH_{1} \gets \cH$. 

\For{$i=1$ \KwTo $|\cC|$}{
\ForEach{\textnormal{response $\cH'$, $\cH'\in t_{i}\br{\cH_{i}}$, such that $1\leq\spr{\cH'} \leq \frac{5\beta-1}{5\beta}\cdot \spr{\cH}$}}{
    $D_{\cH'} \gets$ \ProcDecisionTree$(\cH', \cT, \preceq, \blackBox)$.

    Hang $D_{\cH'}$ below $t_{i}$ in $D$.
}
\If{\textnormal{there is no $\cH'$, $\cH'\in t_{i}\br{\cH_{i}}$, such that $\spr{\cH'} > \frac{5\beta-1}{5\beta}\cdot \spr{\cH}$}}{
    \Return $D$.
}

$\cH_{i+1} \gets\argmax_{\cH', \cH'\in t_{i}\br{\cH_{i}}} \brc{\spr{\cH'}}$.

$t_{i+1} \gets$ the test associated with $i+1$-th set in $\cC$.

Hang $t_{i+1}$ below $t_{i}$.
}
\Return $D$.
}
\end{algorithm}

It should be noted, that we can restrict any decision tree $D$ to a subset of tests $\cC\subseteq \cT$, such that $D$ restricted to $\cC$ is connected and containing $r\br{D}$ and we can measure the contribution of $\cC$ to the cost of $D$. This is done by restricting each $\tests{h}{D}$ to $\cC$ and summing the costs of all tests in the restricted paths. The two cost criteria are then defined accordingly.
Let $D_{\cC}$ be the decision tree $D$ restricted only to tests corresponding to members of $\cC$, i.~e., we only consider the path formed by these tests.

\subsection{Precedence Constrained Worst Case Decision Tree} \label{subsection:PCWCDT}

We start with the worst-case cost criterion, which is simpler to analyze. We have the following immediate corollary of the definition of sepcover.
\begin{observation} \label{obs:sepcoverWC}
    Let $\cI$ be any instance of $\ProblemPCWCDT$.
    Then, $\spr{\sepcover{\cI}} \leq \OPT\br{\cI}$.
\end{observation}

%We will use $\spr{\sepcover{\cI}}$ as a lower bound on $\OPT\br{\cI}$ in the analysis of the approximation algorithm for $\ProblemPCWCDT$.
\begin{lemma}\label{lemma:cover_sep_worst_case}
    Let $\cI=\br{\cH, \cT, \preceq}$ be any $\ProblemPCWCDT$ instance.
    Let $\cC^*$ be an optimal solution to $\ProblemPCSC$ on instance $\br{\cU, \xi(\cT), \preceq_{\xi}, f=1/4}$.
    %and a test $t$ covers $h\in \cH$ if $\brc{U_{t}\br{u}}\leq \frac{3}{4}\cdot \brc{\mathcal{U}}$.
    Then, $\spr{\cC^*} \leq \spr{\sepcover{\cI}}$.
\end{lemma}
    \begin{proof}
        It is enough to argue that $\xi(\sepcover{\cI})$ covers at least $\spr{\cU/4}$ elements from $\cU$.
        Assume towards a contradiction that this is not the case.
%Therefore there exists $t\in \sepcover{\cI}$ and a reply $\cH'$ of size $\spr{\cH'}\leq \spr{\cH}/2$ such that hypotheses in $\cH'$ are not covered by $\sepcover{\cI}$, otherwise the claim holds trivially, since all hypotheses are covered. Let $\cH'\subseteq U_{t,j}$ (since $\cH'$ is a reply to a test $t$, such $U_{t,j}$ always exists). By assumption, we have that $\spr{U_{t,j}-\cH'} < \spr{\cH}/4$. Therefore, we have that $\spr{U_{t,j}} = \spr{\cH'}+\spr{U_{t,j}-\cH'}<3/4\cdot \spr{\cH}$ which by definition means that $h$ is covered by $\sepcover{\cI}$, a contradiction.
        Therefore, there exists a test $t\in \sepcover{\cI}$ which is executed first in $D^*$ and a reply $H\in t$ that corresponds to a child $u$ of $t$ such that hypotheses in $H$ are not covered by $\xi(\sepcover{\cI})$ and $\spr{\cU_u}\leq\spr{\cU}/2$.
        Otherwise the claim holds trivially, since all hypotheses would be covered.
        Note that $\cU_u\subseteq H$.
        We have that $H\setminus\cU_u$ is a subset of hypotheses covered by other tests in $\sepcover{\cI}$ and hence by assumption, $\spr{H\setminus\cU_u}< \spr{\cU}/4$.
        Hence, we get $\spr{H} = \spr{\cU_u}+\spr{H\setminus\cU_u}<3/4\cdot \spr{\cU}$ which means that hypotheses in $H$ are covered by $\sepcover{\cI}$, a contradiction.
    \end{proof}

%\DD{I tak bedziemy oszczedzac miejsce, wiec moze nie warto parafrazowac algorytmu, a dac jakas intuicje jako klej?}
%We now state an algorithm (Algorithm~\ref{alg:worstDecisionTree}) that outputs a decision tree for an input instance $\cI=\br{\cH, \cT, \preceq}$ of $\ProblemPCWCAL$.

%     Algorithm \ref{alg:worstDecisionTree} is recursive and works as follows.
%     Given an instance $\cI=\br{\cH, \cT, \preceq}$, if $\spr{\cH}=1$ we return the trivial decision tree with a single leaf corresponding to the only hypothesis in $\cH$. Otherwise, we run the $\br{\gamma,\alpha}$-approximation algorithm for $\ProblemPCSC$ on instance $\br{\cH, \xi(\cT), \preceq_{\xi}, \spr{\cH}/4}$.
%     Let $S\subseteq\cT$ be the output. We build a decision tree $D_S$ on tests from $S$ closed under $\preceq_{\xi}$.
%     For each $\cH' \in \cH\setminus S$, we recursively call \ProcWorstDecisionTree on instance $\br{\cH', \cT\setminus S, \preceq-S}$ and attach the returned decision tree to the leaf of $D_S$ corresponding to $\cH'$. Finally, we return the constructed decision tree $D$.

The following observation is due to Lemmas \ref{lemma:subspace_opt} and \ref{lemma:cover_sep_worst_case}:
    \begin{observation} \label{obs:DS}
        Consider $D_{\cC}$ computed in Algorithm~\ref{alg:worstDecisionTree}, that uses a $\br{\gamma,\alpha}$-approximation algorithm for $\ProblemPCSC$ as the $\blackBox$ subroutine.
        %Let $D_S$ be the decision tree built on tests from $S$ respecting the precedence constraints $\preceq$.
        Then, $\COSTW\br{D_{\cC}} \leq \alpha \cdot \spr{\sepcover{\cI}}$.
    \end{observation}

\begin{theorem} \label{thm:PCSC-to-PCWCDT}
If there is a polynomial-time bicriteria $\br{\alpha, \beta}$-approximation algorithm for $\ProblemPCSC$, then there is a polynomial-time
$\cO\br{\frac{\alpha}{\log\br{\frac{5\beta}{5\beta -1}}} \cdot \log n}$-approximation algorithm for $\ProblemPCWCDT$. 
% In particular when $\gamma = \cO\br{1}$, the approximation is $\cO\br{\alpha \cdot \log n}$.
\end{theorem}
\begin{proof}
     Consider a call to \ProcDecisionTree on input $\cI=\br{\cH, \cT, \preceq}$, where $\blackBox$ is a $\br{\alpha,\beta}$-approximation algorithm for $\ProblemPCSC$.
     We prove by induction on $n$ that
     \begin{equation} \label{eq:D}
     \COSTW\br{D} \leq \frac{\alpha}{\log\br{\frac{5\beta}{5\beta -1}}}\cdot \log n \cdot \OPT\br{\cI}.
     \end{equation}

     The base case when $n=1$ is trivial since the the cost of the decision tree is 0.
     Assume that for every $\cI' = \br{\cH', \cT, \preceq}$ such that $\spr{\cH'}<\spr{\cH}$ and $n' = \spr{\cH'}$ we have \[
     \COSTW\br{D'} \leq \frac{\alpha}{\log\br{\frac{5\beta}{5\beta -1}}}\cdot \log n' \cdot \OPT\br{\cI'}
     \] 
     where $D'$ is the decision tree recursively returned by \ProcDecisionTree on input $\cI'$.
    %  To point out the dependency of $D'$ on $\cH'$, we write $D'(\cH')$.
     We have
%         \begin{align*}
%             \COSTW\br{D, \cI} \leq & \COSTW\br{D_S, \cI} + \max_{D'(\cH')} \COSTW\br{D', \cI'} \\
%             \leq & \alpha \cdot \spr{S^*} + \max_{D'(\cH')} \frac{\alpha}{\log\br{\frac{4\gamma}{4\gamma -1}}}\cdot \log n' \cdot \OPT\br{\cI'} \\
%             \leq & \alpha \cdot \spr{\sepcover{\cI}} + \frac{\alpha}{\log\br{\frac{4\gamma}{4\gamma -1}}}\cdot \log \br{\frac{\br{{4\gamma-1}}\cdot n}{4\gamma}} \cdot \OPT\br{\cI} \\
%             = & \alpha \cdot \OPT\br{\cI} + \frac{\alpha}{\log\br{\frac{4\gamma}{4\gamma -1}}}\cdot \log n \cdot \OPT\br{\cI} - \alpha\cdot \OPT\br{\cI} \\
%             = & \frac{\alpha}{\log\br{\frac{4\gamma}{4\gamma -1}}}\cdot \log n \cdot \OPT\br{\cI} \\
%        \end{align*}
        \begin{align*}
            \COSTW\br{D} \leq & \COSTW\br{D_{\cC}} + \max_{\cH\in D\setminus D_{\cC}} \brc{\COSTW\br{D'}} \\
            \leq & \alpha \cdot \spr{\sepcover{\cI}} + \max_{\cH\in D\setminus D_{\cC}} \brc{\frac{\alpha}{\log\br{\frac{5\beta}{5\beta -1}}}\cdot \log n' \cdot \OPT\br{\cI'}} \\
            \leq & \alpha \cdot \OPT\br{\cI} + \frac{\alpha}{\log\br{\frac{5\beta}{5\beta -1}}}\cdot \log \br{\frac{\br{{5\beta-1}}\cdot n}{5\beta}} \cdot \OPT\br{\cI} \\
            = & \alpha \cdot \OPT\br{\cI} + \frac{\alpha}{\log\br{\frac{5\beta}{5\beta -1}}}\cdot \log n \cdot \OPT\br{\cI} - \alpha\cdot \OPT\br{\cI} \\
            = & \frac{\alpha}{\log\br{\frac{5\beta}{5\beta -1}}}\cdot \log n \cdot \OPT\br{\cI},
        \end{align*}
        where the second inequality is by the induction hypothesis and Observation~\ref{obs:DS}, the third inequality follows by Observation~\ref{obs:sepcoverWC}, Lemma \ref{lemma:subspace_opt} and the fact that $n'=\spr{\cH'} \leq \br{1-\frac{4}{5}\cdot\frac{1}{4\cdot\beta}}\cdot n=\frac{(5\beta -1)}{5\beta}\cdot n$ for every for every $\cH'\in D\setminus D_{\cC}$.
        This concludes the proof of the theorem.
\end{proof}

Therefore, by Theorem \ref{thm:PCSC-to-PCWCDT} we obtain the following corollary:
\begin{corollary}\label{cor:PCWCDT-special-cases}
    There exist polynomial-time algorithms for $\ProblemPCWCDT$ with the following guarantees:
    \begin{itemize}
        \item $\cO\br{\log n}$-approximation for inforests, by Theorem~\ref{thm:BPCSC-inforest}, and Corollary~\ref{cor:BPCSC-to-PCSC}.
        \item $\cO\br{\log^2n}$-approximation for outforests, by and Theorem~\ref{thm:BPCSC-outforest}.
        \item General $\cO\br{\sqrt{m}\cdot\log n}$-approximation, by Theorem~\ref{thm:MDPCStoPCSC}, and Theorem~\ref{thm:MDPCStoPCSC}.
    \end{itemize}
\end{corollary}

\subsection{Precedence Constrained Average Case Decision Tree} \label{subsection:PCACDT}

Here, the analysis becomes more involved, since in order to compare the cost of our solution with the cost of sepcover, rather then just comparing the size of sepcover with the cost of our solution, we need to take into account the contribution of each test to the objective function.
We start with few observations.
%For a decision tree $D$ and its node $v$, let $c\br{D, v} = p(v)\cdot\spr{\cH_v}$.
\begin{observation} \label{obs:dec-tree-averaging}
    If $D$ is a decision tree for a $\ProblemPCACDT$, then $\COSTA\br{D}=\sum_{v\in V(D)}\spr{\cH_v}$.
%    $$
%    \COSTA\br{D} = \sum_{t \in X} p(\cH_t)\cdot\spr{\cH_t},
%    $$
%    where $X$ is the set of all non-leaf vertices of $D$.
\end{observation}
% Here, $D-S\br$ is the collection of subtrees (a forest) obtained by the removal of all nodes in $S$ from $D$.
Observation~\ref{obs:dec-tree-averaging} allows us to decompose the cost criterion as follows.
\begin{observation} \label{obs:WCdecomposition}
    Let $D$ be any decision tree for an instance of $\ProblemPCACDT$ and let $S=\tests{D}{v}$ for any node $v$ of $D$.
    Then,
    \[
    \COSTA\br{D} = \COSTA\br{D_{S}} + \sum_{D'\in D\setminus V\br{D_{S}}} \COSTA\br{D'}.
    \]
%     $$
%     \COSTA\br{D, \cI} = \COSTA\br{S, \cI} + \sum_{D'\in D\setminus S} \COSTA\br{D', \cI}.
%     $$

\end{observation}
\paraTitle{Remark on notation.}
Throughout this subsection, we treat $\sepcover{\cI}$ as a sequence rather than a set, since the order of tests matters for computing the average-case cost. 
While formally $\sepcover{\cI}$ is defined as the set of tests along a path in the decision tree, when we write $\COSTA\br{\sepcover{\cI}}$ we implicitly refer to the natural ordering of these tests induced by their position on the path from the root.
As an immediate corollary, by taking an optimal $D^*$ and the corresponding $\sepcover{\cI}$ we have:
\begin{observation} \label{obs:sepcoverAC}
    Let $\cI$ be any instance of $\ProblemPCACDT$.
    Then, $\COSTA\br{\sepcover{\cI}} \leq \OPT\br{\cI}$.
\end{observation}
We follow a similar idea as the one for the worst case cost, by using the connection to $\ProblemPCMSSC$ instead of $\ProblemPCSC$.
This allows to use $\COSTA\br{\sepcover{\cI}}$ as a lower bound on $\OPT\br{\cI}$ in the analysis of the approximation algorithm for $\ProblemPCACDT$.
We have a lemma that is analogous to Lemma \ref{lemma:cover_sep_worst_case}.
\begin{lemma}\label{lemma:cover_sep_average_case}
    Let $\cI=\br{\cH, \cT, \preceq}$ be any $\ProblemPCACDT$ instance and let $\cI_{\cU}=\br{\cU, \cT, \preceq}$. Let $\cC^*$ be an optimal solution to $\ProblemPCMSSC$ on instance $\cI_{\cU}$.
    %, where a test $t$ covers element $u \in \cH$ if for $u\in U_{t,j}$, $\spr{U_{t,j}}\leq \frac{3}{4}\cdot \spr{\cH}$.
    Then, $\coverageTime\br{\cC^*} \leq \COSTA\br{\sepcover{\cI_{\cU}}}$.
\end{lemma}
\begin{proof}
    Assume towards a contradiction that $\coverageTime\br{\cC^*} > \COSTA\br{\sepcover{\cI_{\cU}}}$.
    We will show that in such case there exists a cover $\sigma \subseteq \xi(\sepcover{\cI_{\cU}})$ such that $\coverageTime\br{\sigma} \leq \COSTA\br{\sepcover{\cI_{\cU}}}$, contradicting the optimality of $\cC^*$.
    Let $\sigma$ be the shortest prefix of $\sepcover{\cI_{\cU}}$ that covers at least $\spr{\cU}/4$ items.
    Note that because we consider a prefix, the precedence constraints are satisfied.
    Such a subsequence exists, by repeating the argument used in the proof of Lemma \ref{lemma:cover_sep_worst_case}.
    Take an arbitrary $h \in \cU$.
    There are two cases to consider:
    \begin{itemize}
        \item $\sigma$ covers $h$. Consider the first test $t$ that sepcovered $h$ in $\sepcover{\cI_{\cU}}$. By definition, the tests prior to $t$ in $\sigma$ cover at most $\spr{\cU}/4$ items. Since at the moment of sepcovering, $h$ belonged to a reply of size at most $\spr{\cU}/2$, we know that $t$ also covers $h$ in $\sigma$. This means that the contribution of $h$ to $\coverageTime\br{\sigma}$ is at most its contribution to $\COSTA\br{\sepcover{\cI_{\cU}}}$.
        \item $\sigma$ does not cover $h$. Since $h$ is sepcovered by some test $t$ in $\sepcover{\cI_{\cU}}$ but not covered by $\sigma$, the contribution of $h$ to $\coverageTime\br{\sigma}$ is $\spr{\sigma}$ and its contribution to $\COSTA\br{\sepcover{\cI_{\cU}}}$ is at least $\spr{\sigma}$.
    \end{itemize}
    Thus, we have that $\coverageTime\br{\sigma} < \COSTA\br{\sepcover{\cI_{\cU}}}$, a contradiction.
\end{proof}

The following observation is analogous to Observation~\ref{obs:DS} however requires a more careful consideration.
\begin{observation} \label{obs:DS-average}
    Consider $D_{\cC}$ computed in Algorithm~\ref{alg:worstDecisionTree}, that uses an $\br{\alpha,\beta}$-approximation algorithm for $\ProblemPCMSSC$ as the $\blackBox$ subroutine.
    Then, $\COSTA\br{D_{\cC}} \leq 2\alpha \cdot \COSTA\br{\sepcover{\cI_{\cU}}}$.
\end{observation}
\begin{proof}
    Firstly, examine how the Algorithm \ref{alg:worstDecisionTree} constructs $D_{\cC}$. It builds a decision tree by obeying the ordering of tests in $\cC$ and whenever a response is small enough the rest of the decision tree is built recursively. We want to have $\COSTA\br{D_{\cC}}\leq \coverageTime\br{\cC}+\spr{\cC}$. To do so, we show that for $i\in [\cC]$ and hypothesis $h\in\cU$ whenever $h\in\cov\br{\prefix{\cC}{i}}$, $h\notin \cH_{i+1}$. This would mean that the contribution of $h$ to $\COSTA\br{D_{\cC}}$ is at most its contribution to $\coverageTime\br{\cC}$. To see that this is the case, observe that $h\in\cov\br{\prefix{\cC}{i}}$ means that there exists a test $t\in \cov\br{\prefix{\cC}{i}}$, such that $h\in H$, for some reply $H\in t$, such that $\spr{H}\leq \frac{3}{4}\cdot\spr{\cH}$. However, this also mean that the response to this test containing $t$ is of size at most $\frac{3}{4}\cdot\spr{\cH}\leq \frac{5\beta-1}{5\beta}\cdot \spr{\cH}$, and hence $h\notin \cH_{i+1}$ as required. 
    
    Moreover if the single hypothesis $h\in \cH\setminus\cU$ exists the contribution of $h$ to $\COSTA\br{D_{\cC}}$ is at most $\spr{\cC}$. We also trivially have that $\spr{\cC}\leq \coverageTime\br{\cC}$. Hence, we get that 
    \[
    \COSTA\br{D_{\cC}} \leq \coverageTime\br{\cC} + \spr{\cC} \leq 2\cdot \coverageTime\br{\cC}\leq 
    2\alpha \cdot \COSTA\br{\sepcover{\cI_{\cU}}}
    \]
\end{proof}

\begin{theorem}\label{thm:PCMSSC-to-PCACDT}
    If there is a polynomial-time bicriteria $\br{\alpha, \beta}$-approximation algorithm for $\ProblemPCMSSC$, then there is a polynomial-time $\cO\br{\frac{\alpha}{\log\br{\frac{5\beta}{5\beta -1}}}\cdot \log n}$-approximation algorithm for $\ProblemPCACDT$.
    % In particular when $\mu = \cO\br{1}$, the ratio is $\cO\br{\beta \cdot \log \br{m+n}}$.
\end{theorem}
\begin{proof}
    Consider a call to \ProcDecisionTree on input $\cI=\br{\cH, \cT, \preceq}$ where $\blackBox$ is an $\br{\alpha,\beta}$-approximation algorithm for $\ProblemPCMSSC$.
    We argue by induction on $n$ that
    \begin{equation} \label{eq:D-average}
     \COSTA\br{D} \leq \frac{2\alpha}{\log\br{\frac{5\beta}{5\beta -1}}}\cdot \log n \cdot \OPT\br{\cI}.
    \end{equation}

    The base case when $n=1$ is trivial since the cost of the decision tree is 0.
    Assume by induction that for every $\cI' = \br{\cH', \cT, \preceq}$ such that $\spr{\cH'}<\spr{\cH}$ and $n' = \spr{\cH'}$ we have 
    \[
    \COSTA\br{D'} \leq \frac{2\alpha}{\log\br{\frac{5\beta}{5\beta -1}}}\cdot \log n' \cdot \OPT\br{\cI'}
    \]
    where $D'$ is the decision tree recursively returned by \ProcDecisionTree on input $\cI'$.
    % To point out the dependency of $D'$ on $\cH'$, we write $D'(\cH')$.
    We have
\begin{align*}
    \COSTA\br{D} \leq & \COSTA\br{D_{\cC}} + \sum_{\cH\in D\setminus D_{\cC}} \COSTA\br{D'} \\
    \leq & 2\alpha \cdot \COSTA\br{\sepcover{\cI_{\cU}}} + \sum_{\cH\in D\setminus D_{\cC}} \frac{2\alpha}{\log\br{\frac{5\beta}{5\beta -1}}}\cdot \log n' \cdot \OPT\br{\cI'} \\
    \leq & 2\alpha \cdot \OPT\br{\cI} + \frac{2\alpha}{\log\br{\frac{5\beta}{5\beta -1}}}\cdot \log \br{\frac{\br{{5\beta-1}}\cdot n}{5\beta}} \cdot \OPT\br{\cI} \\
    = & 2\alpha \cdot \OPT\br{\cI} + \frac{2\alpha}{\log\br{\frac{5\beta}{5\beta -1}}}\cdot \log n \cdot \OPT\br{\cI} - 2\alpha\cdot \OPT\br{\cI} \\
    = & \frac{2\alpha}{\log\br{\frac{5\beta}{5\beta -1}}}\cdot \log n \cdot \OPT\br{\cI},
\end{align*}
where the first inequality is by Observation~\ref{obs:WCdecomposition}, the second follows by the induction hypothesis and Observation~\ref{obs:DS-average}, the third follows by Observation~\ref{obs:sepcoverAC}, Lemma~\ref{lemma:subspace_opt} and the fact that $n'=\spr{\cH'} \leq \br{1-\frac{4}{5}\cdot\frac{1}{4\cdot\beta}}\cdot n=\frac{(5\beta -1)}{5\beta}\cdot n$ for every $\cH'\in D\setminus D_{\cC}$.
This concludes the proof of the theorem.
\end{proof}

Thus, by Theorem \ref{thm:PCMSSC-to-PCACDT} we obtain the following corollary:

\begin{corollary}\label{cor:PCACDT-special-cases}
    There exist poly-time algorithms for $\ProblemPCACDT$ with the following guarantees:
    \begin{itemize}
        \item $\cO\br{\log n}$-approximation for inforests, by Corollary~\ref{cor:BPCSC-inforest-to-PCMSSC}.
        \item $\cO\br{\log^2 n}$-approximation for outforests, by Theorem~\ref{thm:BPCSC-outforest}.
        \item $\cO\br{\sqrt{m}\cdot\log^{3/2} n}$-approximation for general precedence constraints, by Theorem~\ref{thm:BPCSC-to-PCMSSC} and Theorem~\ref{thm:BPCSC}.
    \end{itemize}
\end{corollary}

\section{Max-Density Precedence-Closed Subfamily} \label{sec:MDPCS}

In order to solve $\ProblemPCSC$ and $\ProblemPCMSSC$, which is a requirement for our main results, we firstly solve two essential subproblems: $\ProblemMDPCS$ and $\ProblemBMDPCS$. We will then use such subroutines as an oracle in order to construct greedy, set-cover like approximation algorithms for our problems.
%An approximation algorithm for $\ProblemMDPCS$ can be used as an essential subroutine in our algorithms for PCSC and PCMSSC.
By \cite{PCMSSC}, the following solution to $\ProblemMDPCS$ on instance $\br{\cU,\cS,\preceq}$ provides an $\bigo\br{\sqrt{m}}$-approximation: pick $\cC$ to be the best solution among
\[
\argmax_{\closure{S}, S\in\cS} \brc{\Delta\br{\closure{S}}}
\]
and $\cS$.
We will refer to such $\cC$ as a \emph{greedy} solution.

When $\max_{S \in \cS} \Delta\br{\closure{S}} \geq 1$, then the approximation factor of the greedy can also be bounded by $\bigo\br{\sqrt{n}}$.
Additionally, we show that a slightly modified greedy rule achieves an $\budget$-approximation for the parametrized version of the problem, $\ProblemBMDPCS$.
Any family $\cC$ that satisfies
    \[
    \cC=\argmax_{\closure{S}, S \in \cS, \spr{\closure{S}}\leq \budget} \brc{\Delta\br{\closure{S}}}
    \]
will be called a \emph{$\budget$-greedy} solution to $\ProblemBMDPCS$.
\begin{theorem}\label{thm:BMDPCS-greedy}
    Let $\cC^*$ be an optimal solution to $\ProblemBMDPCS$ and let $\cC$ be $\budget$-greedy.
    Then, $\Delta\br{\cC} \geq \frac{\Delta\br{\cC^*}}{\budget}$.
\end{theorem}
\begin{proof}
    We utilize a partial argument of \cite{PCMSSC} for $\ProblemMDPCS$. Let $k=\spr{\cC^*}\leq \budget$.
    For each $S\in \cC^*$, by the greedy rule we trivially have that $\Delta\br{\closure{S}}\leq \Delta\br{\cC}$.
    Therefore:
    \[
    \spr{\cov\br{\cC^*}} \leq \sum_{S\in\cC^*}\spr{\cov\br{\closure{S}}}\leq \Delta\br{\cC}\sum_{S\in\cC^*} \spr{\closure{S}} \leq \Delta\br{\cC} \cdot k^2,
    \]
    where the first inequality is by the definition of union. We have that:
    \[
    \frac{\Delta\br{\cC^*}}{\Delta\br{\cC}}=\frac{\spr{\cov\br{\cC^*}}}{k\cdot\Delta\br{\cC}}\leq k\leq \budget,
    \]
    which by rearranging gives the desired inequality.
\end{proof}
The above theorem gives a $\budget$-approximation algorithm for the $\ProblemBMDPCS$. Note that for large values of $\budget$ this might be much worse than the $\cO\br{\sqrt{m}}$-approximation provided for the non-parametrized version of the problem. In the original version of the above argument it could be shown that if one additionally considers a candidate solution which is the whole $\cS$, then the approximation of greedy is $\min\brc{k, m/k}=\sqrt{m}$. However, since we enforce additional condition on the size of the dense subfamily this might not be a feasible solution. Luckily, it turns out that the $\budget$-approximation will be sufficient for our further needs. We leave as an open question whether one can obtain a $o\br{m}$-approximation for $\ProblemBMDPCS$.

% Also, note that the above analysis works also when we associate non-uniform weights to elements of $\cU$ and define density with respect to the total weight of covered elements.

\section{Set covering with precedence constraints} \label{sec:SC}

% Similarly as in previous section we assume our input instance to be unweighted. Note that the analysis of the following procedures works also in the weighted setting, where every element $u\in\cU$ has a weight $w(u)$ and the goal is to maximize the total weight of covered elements or cover at least a given total weight of elements with minimal size of the cover. Since this generalization is straightforward we omit it for clarity.

Having obtained appropriate approximation algorithms for $\ProblemMDPCS$ and $\ProblemBMDPCS$, we now turn our attention towards approximating various set covering problems. In particular, we show how to obtain a bicriteria $\br{\bigo\br{\sqrt{m\cdot \log n}}, 1}$-approximation for $\ProblemBPCSC$. This can be then turned into a $\br{\bigo\br{\sqrt{m\cdot \log n}}, 1}$-approximation for $\ProblemPCSC$ and a $\br{\bigo\br{\sqrt{m\cdot \log n}}, O\br{1}}$-approximation for $\ProblemPCMSSC$. At the end of the section we also show several better guarantees for inforest and outforest precedence constraints.
\subsection{$\ProblemBPCSC$ via $\ProblemBMDPCS$} \label{subsection:BPCSC}

We start with the budgeted variant, which will become a useful primitive for the fractional variants.
The approximation algorithm is the greedy procedure shown in Algorithm~\ref{alg:BPCSC}.
For the analysis, let $\cC_i$ and $\cA_i$, $i\in[l]$, be the sets $\cC$ and $\cA$ obtained in the $i$-th iteration; $\cC_0=\emptyset$.
Note that $\cov\br{\cA_i,\cU\setminus\cov\br{\cC_{i-1}}}$ is the set of element in $\cU$ that are being covered in the $i$-th iteration. %, and for each $u\in\cov\br{\cA_i,\cU\setminus\cC_{i-1}}$ denote $c(u):=\Delta\br{\cA_i,\cU\setminus\cC_{i-1}}$.
\begin{algorithm}[htb!]
\DontPrintSemicolon
\LinesNumbered

\caption{The $\br{\sqrt{m\cdot H_n}+1, 1}$-approximation algorithm for $\ProblemBPCSC$}
\label{alg:BPCSC}
\SetKwProg{Proc}{procedure}{}{}
\Proc{$\ProcedureBPCSC(\cU, \cS, \preceq, \budget)$}{
%    \If{$\budget\geq \sqrt{m/H_n}$}
%    {
%        \Return $\cS$.
%    }
%  \Else{
  $\cC \gets \emptyset$.

  \While{$\spr{\cC} < \min\brc{\sqrt{m\cdot H_n}\cdot \budget, m}$}{

    $\cA \gets$ Run the $\budget$-approx. algorithm for $\ProblemBMDPCS$ on $(\cU \setminus \cov(\cC), \preceq, \cS\setminus\cC, \budget)$.

%    \ForEach{$u \in \cov\br{\cA, \cU-\cC}$}{
%        $c\br{u}\gets \Delta\br{\cA, \cU-\cC}$. \tcp*{for analysis only}
%    }
    $\cC \gets \cC \cup \cA$.
  }

  \Return $\cC$.
%  }
}
\end{algorithm}

First we bound the size of the set returned by the algorithm (cf. Lemma~\ref{lem:BPCSC-cost}) and then we estimate how much it covers (cf. Lemma~\ref{lem:BPCSC-coverage}).
\begin{lemma}\label{lem:BPCSC-cost}
    Let $\cC$ be the cover returned by Algorithm~\ref{alg:BPCSC}. Then, $\spr{\cC} \leq \br{\sqrt{m\cdot H_n}+1}\cdot \budget$.
\end{lemma}
\begin{proof}
    We consider two cases:
    \begin{enumerate}
        \item If $\budget\geq \sqrt{m/H_n}$, then for $\cC=\cS$ it holds: $\spr{\cC} = m \leq \sqrt{m\cdot H_n}\cdot \budget$.
        \item Else, when $\budget < \sqrt{m/H_n}$, we have that before the last iteration, that is in iteration $l-1$, of the while loop, $\spr{\cC_{l-1}} < \sqrt{m\cdot H_n}\cdot \budget$. Since in the last iteration we add to $\cC_{l-1}$ sets in a collection $\cA_l$ such that $\spr{\cA_l} \leq \budget$, we have that 
        \[
        \spr{\cC_l} \leq \sqrt{m\cdot H_n}\cdot \budget + \budget = \br{\sqrt{m\cdot H_n}+1}\cdot \budget.
        \]
    \end{enumerate}
\end{proof}

\begin{lemma}\label{lem:BPCSC-coverage}
    Let $\cC$ be the cover returned by Algorithm~\ref{alg:BPCSC}. Then, $\spr{\cov\br{\cC}} \geq \spr{\cov\br{\cCopt}}$.
\end{lemma}
\begin{proof}
    If $\budget \geq \sqrt{m/H_n}$, then the algorithm returns $\cC_l=\cS$ and trivially covers at least as many elements as $\cCopt$.

    Otherwise, we assume that the while loop is being executed until the number of covered elements is at least $\spr{\cov\br{\cCopt}}$ and we bound the cost of the cover constructed up to that point by $\br{\sqrt{m\cdot H_n}+1}\cdot \budget$. By doing so, we show that when the cost of the cover $\cC$ exceeds $\sqrt{m\cdot H_n}\cdot \budget$, the number of covered elements is at least $\spr{\cov\br{\cCopt}}$.
    
    %Let $\cC_0 = \emptyset$ and let $\cC_i$ be the cover after the $i$-th iteration of the while loop, for $i \geq 1$.
    Let $R_i = \cU \setminus \cov\br{\cC_i}$ be the set of uncovered elements after the $i$-th iteration; take $R_0=\emptyset$.
    %Let $\cA_i$ be the set added to $\cC_{i-1}$ in the $i$-th iteration, so that $\cC_i = \cC_{i-1} \cup \cA_i$.
    We have that $\spr{\cC_i} = \sum_{j=1}^{i} \spr{\cA_j}$. For any covered element $u \in \cU$, let $i(u)$ be the first iteration in which $u$ is covered, i.e., $u \in \cov\br{\cA_{i(u)}, R_{i(u)-1}}$. We set the \emph{price} of $u$ to be 
    \[c(u) = \frac{\spr{\cA_{i(u)}}}{\spr{\cov\br{\cA_{i(u)}, R_{i(u)-1}}}}.
    \]

    Let $t$ be the index of the first iteration of the while loop when $\spr{\cov\br{\cC_{t}}}\geq\spr{\cov\br{\cCopt}}$.
    (Note that such $t$ exists since we are analyzing Algorithm~\ref{alg:BPCSC} under the assumption that the loop works indefinitely, until the number of covered elements is large enough).
    Let $k=\spr{\cov\br{\cC_{t-1}}}\leq\spr{\cov\br{\cCopt}}$.
    Order the elements of $\cov\br{\cCopt}$ as $u_1, u_2, \ldots, u_k$ in the order in which they are covered by the algorithm with ties broken arbitrarily (i.e., if $i(u_j)<i(u_{j'})$, then $j<j'$).

    \begin{claim}\label{lem:cost-per-element}
        For each $j \in \brc{1, \ldots, k}$ 
        \[
        c(u_j) \leq \frac{\budget^2}{k-j+1}.
        \]
    \end{claim}
    \begin{proof}
     Consider the iteration $i(u_j)$ in which $u_j$ is covered. Since $\cCopt$ is a precedence-closed family, we know that during iteration $i(u_j)$, there exists a precedence-closed family $\cCopt\cap R_{i(u_j)-1}$ that covers at least $k-j+1$ elements of $R_{i(u_j)-1}$ and has size at most $\spr{\cCopt}\leq \budget$. Since the algorithm selects a $\budget$-approximation to $\ProblemBMDPCS$ during iteration $i(u_j)$, we have that:
    \[
    \Delta\br{\cA_{i(u_j)}, R_{i(u_j)-1}} \geq \frac{\Delta\br{\cCopt, R_{i(u_j)-1}}}{\budget} \geq \frac{k-j+1}{\spr{\cCopt} \cdot \budget} \geq \frac{k-j+1}{\budget^2},
    \] 
    where the first inequality is by the approximation guarantee of the algorithm and the greedy choice, the second inequality is by definition of density, and the last inequality is by the budget constraint on $\cCopt$. Rearranging the above inequality yields:
    \[
    c(u_j) = \frac{\spr{\cA_{i(u_j)}}}{\spr{\cov\br{\cA_{i(u_j)}, R_{i(u_j)-1}}}} = 
    \frac{1}{\Delta\br{\cA_{i(u_j)}, R_{i(u_j)-1}}} \leq \frac{\budget^2}{k-j+1}
    \]
    and the claim follows.
    \end{proof}

    Using Claim~\ref{lem:cost-per-element}, the sum of prices of all elements in $\cC_{t-1}$ is bounded by $\sum_{j=1}^k c(u_j)\leq \budget^2\cdot H_k$.
    By the price definition, $\sum_{j=1}^k c(u_j)=\spr{\cC_{t-1}}$.
    Since $\spr{\cA_t}\leq \budget$, we get 
    \[\spr{\cC_t}\leq\spr{\cC_{t-1}} + \spr{\cA_t} \leq \budget^2\cdot H_n+\budget.
    \]
    Recall that we are considering the case when $\budget\leq\sqrt{m/H_n}$, which gives $\spr{\cC_t}\leq \br{\sqrt{m\cdot H_n}+1}\cdot\budget$.
    This completes the proof of the lemma.
%     \begin{align*}
%         \spr{\cC} &\leq \spr{\cC_{t-1}} + \spr{\cA_t}
%     \\&\leq
%     \sum_{u \in \cov\br{\cC_{t-1}}} c(u) + \budget
%     \\&\leq
%     \sum_{j=1}^{k} \frac{\budget^2}{k-j+1} + \budget
%     \\&=
%     \budget^2 \cdot H_k + \budget
%     \\&\leq
%     \budget^2 \cdot H_n + \budget
%     \\&\leq
%     \br{\sqrt{m\cdot H_n}+1}\cdot \budget,
%     \end{align*}
%     where the first inequality is by construction of $\cC$, the second inequality is by definition of the cost of covered elements and the fact that for every $1\leq i\leq t$, $\spr{\cA}\leq \budget$, the third inequality is by the previous lemma, the fourth equality is by definition of harmonic numbers, the fifth inequality is because trivially $k \leq n$, and the last inequality is by the assumption that $\budget < \sqrt{m/H_n}$.
\end{proof}

Lemmas \ref{lem:BPCSC-cost} and \ref{lem:BPCSC-coverage} prove the following.
\begin{theorem} \label{thm:BPCSC}
    There exists a polynomial-time $\br{\sqrt{m\cdot H_n}+1, 1}$-bicriteria approximation algorithm for $\ProblemBPCSC$, where $n=\spr{\cU}$ and $m=\spr{\cS}$. That is, the algorithm returns a solution $\cC$ such that $\spr{\cov\br{\cC}}\geq \OPT\br{\cI}$ and $\spr{\cC}\leq \br{\sqrt{m\cdot H_n}+1}\cdot \budget$.
\end{theorem}

\subsection{$\ProblemPCSC$ via $\ProblemMDPCS$}\label{subsection:PCSC}
Let $k=f\cdot n$.
In order to prove the main theorem of this section consider a greedy procedure shown as Algorithm~\ref{alg:PCSC}.
We note a subtle fact of finishing the loop once $k/2$ elements are covered and not iterating till reaching $k$.
This is because the last set $\cA$ may be of unbounded size but we argue that the size of an optimal solution to the instance from the last iteration is lower-bounded by $\Omega(k)$, which is enough to obtain the required approximation ratio.
\begin{algorithm}[h]
\DontPrintSemicolon
\LinesNumbered
\caption{A $\br{4\gamma/f, 2}$-approximate greedy algorithm for $\ProblemPCSC$}
\label{alg:PCSC}
\SetKwProg{Proc}{procedure}{}{}
\Proc{$\ProcedurePCSC(\cU, \cS, \preceq, k)$}{
  $\cC \gets \emptyset$\;

  \While{$\spr{\cov\br{\cC}} < k/2$}{

    $\cA \gets$ Run the $\gamma$-approx. algorithm for $\ProblemMDPCS$ on $(\cU \setminus \cov(\cC), \preceq, \cS\setminus\cC)$\;

    $\cC \gets \cC \cup \cA$\;
  }

  \Return $\cC$
}
\end{algorithm}

% \begin{algorithm}[h]
% \DontPrintSemicolon
% \LinesNumbered
% \caption{The $\gamma$-greedy algorithm for PCSC}
% \label{alg:PCSC}
% \SetKwProg{Proc}{procedure}{}{}
% \Proc{\textsc{PCSC}$(\cU, \cS, \mathcal{F}, K)$}{
% $\cC \gets \emptyset$\;
%
% \While{$\spr{\cov\br{\cC, \cU}} < K$}{
% $\cA \gets$ Run the $\gamma$-approx. algorithm for MDPCS on $(\cU - \cC, \cS-\cC, \mathcal{F}-\cC, m)$\;
%
% \If{$\spr{\cov\br{\cC\cup \cA, \cU}} \geq K$}
% {
%     Find the minimum budget $B \in [\spr{\cA}]$, such that the $\gamma$-approx. algorithm for MDPCS on $(\cU - \cC, \cS-\cC, \mathcal{F}-\cC, B)$ returns a set
%     $\mathcal{B}$ with $\spr{\cov\br{\mathcal{B}, \cU-\cC}} \geq \frac{K-\cov\br{\cC, \cU}}{\alpha}$\;
%
%     \While{$\spr{\cov\br{\cC, \cU}} < K$}{
%         $\mathcal{B} \gets$ Run the $\gamma$-approx. algorithm for MDPCS on $(\cU - \cC, \cS-\cC, \mathcal{F}-\cC, B)$\;
%
%         $\cC\gets \cC \cup \mathcal{B}$\;
%     }
%
%     \Return $\cC$
% }
% \ForEach{$u \in \cov\br{\cA, \cU-\cC}$}{
%     $c\br{u}\gets \Delta\br{\cA, \cU-\cC}$\;
% }
% $\cC \gets \cC \cup \cA$\;
% }
% }
% \end{algorithm}

\begin{theorem} \label{thm:MDPCStoPCSC}
    If there exists a polynomial-time $\gamma$-approximation algorithm for $\ProblemMDPCS$, then there exists a polynomial-time bicriteria $\br{4\gamma/f, 2}$-approximation algorithm for $\ProblemPCSC$.
\end{theorem}
\begin{proof}
Denote by $\cA_i$ the set $\cA$ from the $i$th iteration of Algorithm~\ref{alg:PCSC}, $i\in[l]$.
For brevity let $\cC_i$, $i\in[l]$, be the set $\cC$ obtained in the $i$th iteration, $\cC_i=\bigcup_{j=1}^i \cA_j$.
Let $R_i=\cU\setminus\cov\br{\cC_i}$ for $i\in[l]$ and $R_0=\cU$, $\cC_0=\emptyset$.
Denote by $\cI=(\cU,\cS,\preceq,k)$ any input instance to $\ProblemPCSC$, and let $\cCopt$ be an optimal solution to $\ProblemPCSC$ on $\cI$, i.e., $\cov\br{\cCopt}\geq k$ and $\spr{\cCopt}=\optPCSC{\cI}$.

Note that for each $i\in[l]$, $\cCopt\setminus\cC_{i-1}\neq\emptyset$ because otherwise $\cC_{i-1}\subseteq\cCopt$ which means that $\spr{\cC_{i-1}}\geq\spr{\cCopt}\leq k$ contradicting the fact that the algorithm conducted the $l$-th iteration.
Hence, $\cCopt\setminus\cC_{i-1}$ is a precedence-closed family of density not greater than the density of an optimal solution to $\ProblemMDPCS$ for the input provided in the $i$-th iteration, $i\in[l]$.
Thus, 
\[
\Delta(\cA_i,R_{i-1}) \geq \frac{\Delta(\cCopt\setminus\cC_{i-1},R_{i-1})}{\gamma},
\]
which gives by definition of $\Delta$:
\[
\frac{\spr{\cov\br{\cA_i,R_{i-1}}}}{\spr{\cA_i}} \geq \frac{1}{\gamma}\cdot\frac{\spr{\cov\br{\cCopt,R_{i-1}}}}{\spr{\cCopt\setminus\cC_{i-1}}}, \quad i\in[l],
\]
which further implies:
\begin{equation} \label{eq:Ai}
\spr{\cA_i} \leq \gamma \cdot \frac{ \spr{\cov\br{\cA_i,R_{i-1}}} \cdot \spr{\cCopt\setminus\cC_{i-1}} }{ \spr{\cov\br{\cCopt,R_{i-1}}} }
            \leq \gamma \cdot \frac{ \spr{\cov\br{\cA_i,R_{i-1}}} }{ \spr{\cov\br{\cCopt,R_{i-1}}} } \cdot \spr{\cCopt}.
\end{equation}
For each $i\in[l-1]$ it holds $\spr{\cov\br{\cCopt,R_{i-1}}}\geq k/2$.
Thus,
\[
\sum_{i=1}^{l-1}\spr{\cA_i} \leq \frac{2\gamma}{k} \cdot \spr{\cCopt} \cdot \sum_{i=1}^{l-1}\spr{\cov\br{\cA_i,R_{i-1}}} \leq \frac{2\gamma n}{k} \cdot \spr{\cCopt}.
\]
For the last iteration $\spr{\cov\br{\cA_l,R_{l-1}}}\leq n$ (potentially $\cA_l$ may be of size unbounded by a function of $k$) and $\spr{\cov\br{\cCopt,R_{l-1}}}\geq k/2$.
Hence by \eqref{eq:Ai} we get $\spr{\cA_l}\leq \frac{2\gamma n}{k}\spr{\cCopt}$.
This gives 
\[
\spr{\cC_l}=\sum_{i=1}^l\spr{\cA_i}\leq\frac{4\gamma n}{k}\cdot\spr{\cCopt} = \frac{4\gamma}{f}\cdot\spr{\cCopt}.
\]
Moreover, $\spr{\cov\br{\cC_l}}\geq k/2$, which means the algorithm covers at least half of the required elements.
This proves theorem.
\end{proof}

Additionally, observe that if one has an $\br{\alpha, \beta}$-approximation algorithm for $\ProblemBPCSC$, then one can obtain an $\br{\alpha, \beta}$-approximation algorithm for $\ProblemPCSC$ by simply guessing the optimal budget $\budget^*=\spr{\cCopt}$ and running the $\ProblemBPCSC$ algorithm with budget $\budget^*$. We have the following corollary.
\begin{corollary}\label{cor:BPCSC-to-PCSC}
    If there exists a polynomial-time $(\alpha,\beta)$-approximation algorithm for $\ProblemBPCSC$, then there exists a polynomial-time $(\alpha,\beta)$-approximation algorithm for $\ProblemPCSC$.
\end{corollary}

\subsection{$\ProblemPCMSSC$ via $\ProblemBPCSC$} \label{subsection:PCMSSC}
We show how to use an $(\alpha,\beta)$-approximation algorithm for $\ProblemBPCSC$ to obtain an $\br{\bigo\br{\alpha\cdot\beta^3}, \beta}$-approximation for $\ProblemPCMSSC$. This result seems counterintuitive at first glance, since we are using an algorithm for a problem where the budget denotes the maximal allowed size of the cover to solve a problem where the budget denotes the sum of covering times. Our algorithm and its analysis are a generalization of the approach of \cite{ApproxAlgsForOptDTsAndAdapTSPProblems}. Furthermore, their result is based on techniques used for the \textsc{Minimum Latency Travelling Salesman Problem} \cite{PathsTreesMinimumLatencyTours,KTravelingRepairmenProblem}.

Let $1\geq a=\frac{3\beta-1}{3\beta-2}\leq 2$.
Consider the Algorithm~\ref{alg:fPCMSSC}.
% Denote by $\cCopt$ any optimal solution to $\ProblemPCMSSC$ on the input $(\cU, \cS, \preceq, f)$.
Let $l^*$ be such that $a^{l^*-1} < \spr{\cCopt} \leq a^{l^*}$.
\begin{algorithm}[h]
\DontPrintSemicolon
\LinesNumbered
\caption{The algorithm for $\ProblemPCMSSC$ based on $\br{\alpha, \beta}$-approximation for $\ProblemBPCSC$}
\label{alg:fPCMSSC}
\SetKwProg{Proc}{procedure}{}{}
\Proc{$\ProcedureFPCMSSC(\cU, \cS, \preceq, f)$}{
    $\cQ\gets \emptyset$.

    \For{$1\leq l\leq \cl{\log_a{\spr{\cS}}}$}{
        $U\gets \cU$.

        $\pi_l\gets \emptyset$.

        \For{$1\leq i\leq l$}{
            $\tau_l^i\gets $ Run the $\br{\alpha, \beta}$-approx. algorithm for $\ProblemBPCSC$ on $(U, \preceq, \cS, \budget=a^{i+1})$.

            $U \gets U \setminus \cov(\tau_l^i)$.

            $\pi_l \gets \pi_l \circ \tau_l^i$.
        }
        Extend $\pi_l$ by appending to it arbitrary sets to ensure that $\spr{\pi_l} \geq \min\brc{\alpha\cdot \beta\cdot a^l, m}$.

        $\sigma_l \gets $ Run the $\br{\alpha, \beta}$-approx. algorithm for $\ProblemBPCSC$ on $(\cU, \preceq, \cS, \budget=a^l)$.

        $\cQ \gets \cQ \cup \brc{\pi_l \circ \sigma_l}$.
    }
    \Return $\argmin_{\pi\in \cQ} \brc{\coverageTime\br{\pi,\cU}}$.
}
\end{algorithm}

We start with a lemma that relates the coverage times and the lengths of $\pi_{l^*}$ and $\cCopt$.
    \begin{lemma}\label{lem:pi-spread-cost}
        We have that $\spr{\pi_{l^*}} = \Theta\br{\alpha \cdot \beta}\cdot \spr{\cCopt}$ and $\coverageTime\br{\pi_{l^*},\cU} = \bigo\br{\alpha\cdot\beta^3}\cdot \coverageTime\br{\cCopt,\cU}$.
    \end{lemma}
    \begin{proof}
        We observe that since we used an $\br{\alpha, \beta}$-approximation algorithm for $\ProblemBPCSC$ in each iteration of the inner for loop, we have that for each $1\leq i\leq l^*$, $\spr{\tau_{l^*}^i} \leq \alpha \cdot a^{i+1}$. Moreover, we get that before the extension step of $\pi_{l^*}$, we had that:
        \begin{equation} \label{eq:pi}
        \spr{\pi_{l^*}} \leq \alpha\cdot \sum_{i=1}^{l^*} a^{i+1}\leq \frac{\alpha\cdot a^{l^*+2}}{a-1}\leq \frac{\alpha\cdot a^{2}\spr{\cCopt}}{a-1}\leq4\alpha\cdot \br{3\beta-2}\cdot\spr{\cCopt}
        \end{equation}
        Since in the extension step we ensure that $\spr{\pi_{l^*}} \geq \alpha\cdot \beta\cdot a^{l^*}$, we have that $\spr{\pi_{l^*}} = \Theta\br{\alpha \cdot \beta}\cdot \spr{\cCopt}$ and the first part of the lemma follows.

        For every $i\in [l^*]$, denote $n_i^*=\spr{\cov\br{\pref{\cCopt}{a^{i}}}}$ and note that $n_{l^*}^* = \spr{\cov\br{\cCopt}}$. Accordingly, for every $i\in [l^*]$ let $n_i=\spr{\cov\br{\angl{\tau_{l^*}^1,\dots, \tau_{l^*}^i}}}$. Additionally, define $n_0=n_0^*=0$. We have:
        \begin{align*}
            \coverageTime\br{\pi,\cU}&\leq \sum_{i=1}^{l^*}\br{n_i-n_{i-1}}\cdot \sum_{j=1}^{i}\alpha\cdot a^{j+1} + \br{n-n_{l^*}}\cdot \spr{\pi_{l^*}}
            \\&\leq
            \sum_{i=1}^{l^*}\br{n_i-n_{i-1}}\cdot \frac{\alpha\cdot a^{i+2}}{a-1} + \br{n-n_{l^*}}\cdot \spr{\pi_{l^*}}
            \\&=
            \sum_{i=1}^{l^*}\br{\br{n-n_{i-1}}-\br{n-n_i}}\cdot \frac{\alpha\cdot a^{i+2}}{a-1} + \br{n-n_{l^*}}\cdot \spr{\pi_{l^*}}
            \\&\leq
            \sum_{i=0}^{l^*}\br{n-n_i}\cdot \frac{\alpha\cdot a^{i+3}}{a-1} =: Q.
        \end{align*}
        where the last inequality is due to \eqref{eq:pi}.
        Moreover, we have that:
        \begin{align*}
        2\cdot\coverageTime\br{\cCopt,\cU}&\geq a^{l^*-1}\br{n-n_{l^*-1}^*}+ \sum_{i=1}^{l^*-1}a^{i-1}\br{n_i^*-n_{i-1}^*}+\br{n-n_{l^*}^*}\cdot\spr{\cCopt}
        \\&\geq
        a^{l^*-1}\br{n-n_{l^*-1}^*}+ \sum_{i=1}^{l^*-1}a^{i-1}\br{\br{n-n_{i-1}^*}-\br{n-n_{i}^*}}+\br{n-n_{l^*}^*}\cdot a^{l^*-1}
        \\&=
        \br{n-n_0^*}+\sum_{i=1}^{l^*-1}\br{a^{i}-a^{i-1}}\cdot\br{n-n_{i}^*}+\br{n-n_{l^*}^*}\cdot a^{l^*-1}
        \\&\geq
        \br{1-\frac{1}{a}}\cdot\sum_{i=0}^{l^*}a^{i}\cdot \br{n-n_i^*}.
        \end{align*}

        Consider any iteration $i\in \brc{1,\dots,l^*}$ of the inner for loop. Let $U_i=\cU\setminus \cov\br{\pref{\cCopt}{a^{i-1}}}$ be the set of uncovered elements at the beginning of iteration $i$ which is provided as a part of the $\cI_i$ instance to the $\ProblemBPCSC$ algorithm. We have that $\OPT\br{\cI_i}\geq n_i^*-n_{i-1}$ since $\pref{\cCopt}{a^{i}}$ is a feasible solution to instance $\cI_i$ of $\ProblemBPCSC$ such that $\spr{\cov\br{\pref{\cCopt}{a^{i}}, U_i}}\geq n_i^*-n_{i-1}$.
        Since we used an $\br{\alpha, \beta}$-approximation algorithm for $\ProblemBPCSC$ in iteration $i$, we have that $n_i-n_{i-1}\geq \frac{1}{\beta}\cdot\br{n_i^*-n_{i-1}}$. As a consequence we see that $n-n_i\leq \frac{\beta-1}{\beta}\br{n-n_{i-1}}+\frac{1}{\beta}\br{n-n_{i}^*}$.
        Thus,
        \begin{align*}
            \frac{\br{a-1}\cdot Q}{\alpha}
            &=
            \sum_{i=0}^{l^*}a^{i+3}\cdot \br{n-n_i}
            \\&\leq
            a^3\cdot n + \frac{1}{\beta}\cdot\sum_{i=1}^{l^*}a^{i+3}\cdot \br{n-n_i^*} + \frac{\beta-1}{\beta}\cdot\sum_{i=1}^{l^*}a^{i+3}\cdot \br{n-n_{i-1}}
            \\&\leq
            \frac{2\cdot a^4}{a-1}\cdot\coverageTime\br{\cCopt,\cU}+\frac{\beta-1}{\beta}\cdot\sum_{i=1}^{l^*}a^{i+3}\cdot \br{n-n_{i-1}}
            \\&=
            \frac{2\cdot a^4}{a-1}\cdot\coverageTime\br{\cCopt,\cU}+\frac{\br{\beta-1}\cdot a}{\beta}\cdot\sum_{i=0}^{l^*-1}a^{i+3}\cdot \br{n-n_{i}}
            \\&\leq
            \frac{2\cdot a^4}{a-1}\cdot\coverageTime\br{\cCopt,\cU}+\frac{\br{\beta-1}\cdot a}{\beta}\cdot\br{a-1}\cdot\frac{Q}{\alpha}.
        \end{align*}
        By rearranging the above inequality we obtain:
        \begin{align*}
            Q&\leq \alpha\cdot\coverageTime\br{\cCopt,\cU}\cdot\frac{2\cdot a^4}{\br{a-1}^2\cdot\br{1-\frac{\br{\beta-1}\cdot a}{\beta}}}
            \\&=
            \alpha\cdot\coverageTime\br{\cCopt,\cU}\cdot\frac{2\cdot\br{\frac{3\beta-1}{3\beta-2}}^4}{\br{\frac{3\beta-1}{3\beta-2}-1}^2\cdot\br{1-\frac{\br{\beta-1}\cdot \br{3\beta-1}}{\beta\cdot\br{3\beta-2}}}}
            \\&=
            \alpha\cdot\coverageTime\br{\cCopt,\cU}\cdot\frac{2\cdot \br{3\beta-1}^2\cdot \br{\frac{3\beta-1}{3\beta-2}}^4}{\br{1-\frac{\br{\beta-1}\cdot \br{3\beta-1}}{\beta\cdot\br{3\beta-2}}}}
        \end{align*}
        by our choice of $a$. One can easily check that for $\beta\geq 1$, $\frac{\br{\beta-1}\cdot \br{3\beta-1}}{\beta\cdot\br{3\beta-2}}\leq \frac{\br{\beta-1/3}}{\beta}$ so that:
        \begin{align*}
            Q&\leq
            \alpha\cdot\coverageTime\br{\cCopt,\cU}\cdot\frac{\br{3\beta-1}^2\cdot\br{\frac{2\cdot 3\beta-1}{3\beta-2}}^4}{\br{1-\frac{\br{\beta-1/3}}{\beta}}}
            \\&=
            \alpha\cdot\coverageTime\br{\cCopt,\cU}\cdot 6\beta\cdot \br{3\beta-1}^2\cdot \br{\frac{3\beta-1}{3\beta-2}}^4
            \leq 864\cdot\alpha\cdot\beta^3\cdot\coverageTime\br{\cCopt,\cU}
        \end{align*}
        The lemma follows.
    \end{proof}

    There are two cases to consider. If after extension step of $\pi_{l^*}$ we have that $\spr{\pi_{l^*}} \geq m$ then the solution covers entire universe and the proof is complete. Otherwise we have the following lemmas:
    \begin{lemma}\label{lem:sigma-coverage}
        We have that $\spr{\cov\br{\sigma_{l^*}}} \geq \frac{f\cdot n}{\beta}$ and $\spr{\sigma_{l^*}} = O\br{\alpha}\cdot \spr{\cCopt}$.
    \end{lemma}
    \begin{proof}
        Since the optimal cover $\cCopt$ satisfies $\spr{\cCopt}\leq a^{l^*}$ and covers at least $f\cdot n$ elements, it is a feasible solution to the $\ProblemBPCSC$ instance $\cI'=(\cU,\cS,\preceq,a^{l^*})$ for which the $(\alpha,\beta)$-approximation $\sigma_{l^*}$ has been computed after the inner loop.
        Hence, we know that $\spr{\sigma_{l^*}} = O\br{\alpha}\cdot \spr{\cCopt}$ and $\spr{\cov\br{\sigma_{l^*}}} \geq \frac{f\cdot n}{\beta}$.
    \end{proof}
    \begin{lemma}\label{lem:combined-coverage-cost}
        We have that $\spr{\cov\br{\pi_{l^*} \circ \sigma_{l^*}}}\geq \frac{f\cdot n}{\beta}$ and $\coverageTime\br{\pi_{l^*} \circ \sigma_{l^*},\cU} = O\br{\alpha\cdot\beta^3}\cdot \coverageTime\br{\cCopt,\cU}$.
    \end{lemma}
    \begin{proof}
        Since $\pi_{l^*} \circ \sigma_{l^*}$ covers all of the elements covered by $\sigma_{l^*}$, by Lemma~\ref{lem:sigma-coverage} we have that $
        \spr{\cov\br{\pi_{l^*} \circ \sigma_{l^*}}}\geq \frac{f\cdot n}{\beta}.$
        For each element $u\in\cU$, let $t_u$ denote the index of the first set in $\pi_{l^*}$ that covers $u$ if $u$ is covered by $\pi_{l^*}$, otherwise set $t_u=\spr{\pi_{l^*}}$.
        Lemma \ref{lem:pi-spread-cost} implies that 
        \[
        \coverageTime\br{\pi_{l^*},\cU}=\sum_{u\in\cU} t_u = O\br{\alpha\cdot\beta^3}\cdot \coverageTime\br{\cCopt,\cU}.
        \]
        Observe, that for each $u\in\cov\br{\pi_{l^*}}$, its cover time in $\pi_{l^*} \circ \sigma_{l^*}$ is also $t_u$.
        For each $u\notin\cov\br{\pi_{l^*}}$, by Lemma~\ref{lem:pi-spread-cost}, $t_u=\Omega\br{ \alpha\cdot\beta\cdot a^{l^*}}=\Omega\br{\alpha\cdot a^{l^*}}$ and its cover time in $\pi_{l^*} \circ \sigma_{l^*}$ is at most 
        \[
        \spr{\pi_{l^*} \circ \sigma_{l^*}}=\spr{\pi_{l^*}} + \spr{\sigma_{l^*}} = \bigo\br{\alpha\cdot\beta}\cdot \spr{\cCopt} + \bigo\br{\alpha}\cdot a^{l^*} = \bigo\br{\alpha\cdot\beta}\cdot \spr{\cCopt}.
        \]
        Therefore, its cover time in $\pi_{l^*} \circ \sigma_{l^*}$ is at most a constant factor larger than $t_u$.
        As a consequence, due to Lemma~\ref{lem:pi-spread-cost} we get that $\coverageTime\br{\pi_{l^*} \circ \sigma_{l^*},\cU} = O\br{\alpha\cdot\beta^3}\cdot \coverageTime\br{\cCopt,\cU}$.
    \end{proof}
Lemma \ref{lem:combined-coverage-cost} directly implies the following theorem.
\begin{theorem}\label{thm:BPCSC-to-PCMSSC}
    If there exists a polynomial-time bicriteria $(\alpha,\beta)$-approximation algorithm for $\ProblemBPCSC$, then there exists a polynomial-time bicriteria $\br{\bigo\br{\alpha\cdot\beta^3}, \beta}$-approximation algorithm for $\ProblemPCMSSC$.
\end{theorem}
It should be noted that in all cases when we use Theorem~\ref{thm:BPCSC-to-PCMSSC} we will have $\beta=\bigo\br{1}$, thus yielding a $\br{\bigo\br{\alpha}, \bigo\br{1}}$-approximation algorithm for $\ProblemPCMSSC$.

\subsection{Special cases of precedence constraints} \label{subsection:SpecialCases}
\begin{theorem}\label{thm:BPCSC-inforest}
    There exist polynomial-time bicriteria $\br{1, \frac{e}{e-1}}$-approximation algorithms for $\ProblemBPCSC$ and $\ProblemPCSC$ when the precedence constraints form an inforest.
\end{theorem}
\begin{proof}
    We show how to reduce such an instance to the well-known \textsc{Maximum Coverage} problem (without any precedence constraints) where we allow the non-uniform cost of the sets. In the latter setup including a set $S$ into the cover may impose arbitrary cost $c\br{S}$. For this task a $\br{1, \frac{e}{e-1}}$-approximation algorithm is known \cite{TheBudgetedMaximumCoverageProblem}. Observe that for precedence constraints of the inforest type, for any two sets $S, W\in \cS$ if $\closure{S}\cap\closure{W}\neq \emptyset$, then one is a descendant of the other in the inforest, i.e., either $S\preceq W$ or $W\preceq S$. Thus, we can create an equivalent \textsc{Maximum Coverage} instance in the following way: for each set $S\in \cS$, we create a new set $S'$ in the \textsc{Maximum Coverage} instance, such that $S'=\bigcup_{W\in \closure{S}}W$ and $c\br{S'}=\spr{\closure{S}}$. It is easy to see that any reasonable cover for the \textsc{Maximum Coverage} instance picks at most one representative $S'$ of set $S$ from each tree $T$ in the inforest and this corresponds to picking all of the sets in $\closure{S}$ in the original instance. Thus, this choice corresponds to picking a precedence-closed subtree of this subtree $T$. It is also easy to see, that any precedence-closed subtree of $T$ corresponds to picking a representative $S'$ of some set $S$ in the \textsc{Maximum Coverage} instance. Therefore, any feasible solution to the \textsc{Maximum Coverage} instance corresponds to a feasible solution to $\ProblemBPCSC$ on the original instance with the same cost and coverage, and vice versa. Thus, by Corollary \ref{cor:BPCSC-to-PCSC}, the theorem follows.
\end{proof}

Via a careful examination of Theorem~\ref{thm:BPCSC-to-PCMSSC}, we obtain the following result for $\ProblemPCMSSC$
\begin{corollary}\label{cor:BPCSC-inforest-to-PCMSSC}
    There exists a bicriteria $\br{864\cdot\br{\frac{e}{e-1}}^3+1, \frac{e}{e-1}}\approx\br{3421.7, \frac{e}{e-1}}$-approximation algorithm for $\ProblemPCMSSC$ when the precedence constraints form an inforest running in polynomial time.
\end{corollary}
\begin{theorem}\label{thm:BPCSC-outforest}
    There exist polynomial-time bicriteria $\br{\bigo\br{\log n}, 4}$-approximation algorithms for $\ProblemBPCSC$,  $\ProblemPCSC$ and $\ProblemPCMSSC$ when the precedence constraints form an outforest.
\end{theorem}
\begin{proof}
    We show how to reduce such an instance to the to $\ProblemGSO$ on tree metrics.
    The latter problem admits a $\br{\bigo\br{\log n}, 4}$-approximation algorithm \cite{ApproxAlgsForOptDTsAndAdapTSPProblems}.
    It should be noted that they give a worse approximation ratio of $\br{\bigo\br{\log^2 n}, 4}$, however their algorithm consists of embedding a general metric into a tree metric with $\bigo\br{\log n}$ distortion.
    Since our metric will be a tree metric, we can skip this step and obtain the improved approximation ratio.
    
    We create this metric as follows: we introduce a dummy root node $r$ and for each set $S\in\cS$ we add a representative node $v_S$ connected to its every successor in the outforest with an edge of length~$1$.
    Additionally, we connect each maximal element of the outforest to $r$ with an edge of length~$1/2$. Finally, for each element $u\in\cU$ we create a group $X_u=\brc{v_S \mid u\in S}$ which consists of all representatives of vertices containing $u$. Since the resulting metric is a tree metric, any tour returned by the algorithm traverses each edge either $0$ or $2$ times.
    Moreover, since the tour starts at the root, we have that if a representative of a set $W$ is visited, the representatives of all of its predecessors in the outforest have been visited.
    Now, observe that if we take a tour in this metric and pick the sets corresponding to the representatives visited by the tour, we obtain a precedence-closed family of sets of size equal to the length of the tour. Conversely, if we take a precedence-closed family of sets, we can construct a tour in this metric that visits the representatives of these sets and has length equal to the size of the family. The number of elements covered by the family is equal to the number of groups $X_u$ visited by the tour.
    Therefore, any feasible solution to $\ProblemGSO$ on this metric corresponds to a feasible solution to $\ProblemBPCSC$ on the original instance with the same cost and coverage, and vice versa. Thus, by Corollary \ref{cor:BPCSC-to-PCSC} and Theorem \ref{thm:BPCSC-to-PCMSSC} theorem follows.
\end{proof}

\section{Hardness} \label{sec:hardness}
\newcommand{\coord}[3]{h[#1, #2, #3]}
\newcommand{\sepedge}[1]{t^{sep}[#1]}
\newcommand{\varedge}[1]{t^{var}[#1]}
\newcommand{\litedge}[1]{t^{cla}[#1]}

In this section, we establish hardness and inapproximability results for precedence-constrained covering and decision tree problems introduced in previous sections.
% Throughout all reductions involving fractional set covering, we assume that $f = 1$.

\subsection{Reductions from precedence constrained set cover}
\begin{theorem}\label{thm:PCSC-to-PCWCDT-hardness}
    Approximating $\ProblemPCWCDT$ is as hard as approximating $\ProblemPCSC$ with $f=1$.
\end{theorem}
\begin{proof}
    Let $(\cU, \cS, \preceq)$ be an instance of $\ProblemPCSC$ where $\cU$ is the universe of elements, $\cS$ is the family of sets and $\br{\cS,\preceq}$ are the precedence constraints on $\cS$. We construct an instance of $\ProblemPCWCDT$ as follows. For each element $u \in \cU$, we create two representative hypotheses $h_u^0$ and $h_u^1$. 
    % \DD{tutaj chyba należałoby powiedzieć która hipoteza jest targetem w których sytuacjach?} 
    For each set $S \in \cS$ and $i\in\brc{0,1}$, we create a test $t_S$ that gives a reply $\brc{h_u^i}$  if $u \in S$ and all other hypotheses belong to a reply $H_{S}$ when $u\notin S$. The precedence constraints on tests are the same as the precedence constraints on sets, i.e., if $S_1$ must be selected before $S_2$ in $\ProblemPCSC$, then $t_{S_1}$ must be selected before $t_{S_2}$ in $\ProblemPCWCDT$. Notice, that any valid decision tree for the constructed instance of $\ProblemPCWCDT$ is a path of tests (ignoring leaves) since performing any test $t_S$ either returns a response containing one of the hypotheses $h_u^i$, thus ending the identification process or returns $H_{S}$, which is a singular possible response leading to the next test. Even when the candidate set contains only representatives of one element $u$, we still need to perform a test to distinguish between $h_u^0$ and $h_u^1$, (which could not be a case if there was only one representative hypothesis per element).
    Therefore, any valid decision tree for the constructed instance of $\ProblemPCWCDT$ corresponds to a valid selection of sets in $\ProblemPCSC$ and vice versa. Moreover, the cost of the decision tree is equal to the number of selected sets. Thus, since the size of the reduction is linear any $\alpha$-approximation for $\ProblemPCWCDT$ would yield an $\alpha$-approximation for $\ProblemPCSC$.
\end{proof}
\begin{theorem}\label{thm:PCMSSC-to-PCACDT-hardness}
    Approximating $\ProblemPCACDT$ is up to a constant factor, as hard as approximating $\ProblemPCMSSC$ with $f=1$.
\end{theorem}
\begin{proof}
    We use the same reduction as in the previous theorem. Note that since we doubled the number of hypotheses in the construction, the cost of any decision tree in the constructed instance for $\ProblemPCACDT$ has twice the number of selected sets in $\ProblemPCMSSC$. However, this does not affect the structure of the reduction and thus any $\alpha$-approximation for $\ProblemPCACDT$ would yield an $\cO\br{\alpha}$-approximation for $\ProblemPCMSSC$.
\end{proof}

\subsection{Outforest precedence constraints}
\begin{theorem}\label{thm:PCSC-outforest-hardness}
    $\ProblemPCSC$ with $f=1$ and outforest precedence constraints cannot be approximated within a factor of $\cO\br{\log^{2-\epsilon} n}$ for any $\epsilon > 0$ unless $\text{NP}\subseteq \text{ZTIME}\br{n^{\text{polylog}(n)}}$.
\end{theorem}
\begin{proof}
    We show that $\ProblemGST$ on tree metrics is reducible to $\ProblemPCSC$ with outforest precedence constraints. The  $\ProblemGST$ on trees cannot be approximated within a factor of $\cO\br{\log^{2-\epsilon} n}$ for any $\epsilon > 0$ unless $\text{NP}\subseteq \text{ZTIME}\br{n^{\text{polylog}(n)}}$ \cite{PolylogarithmicInapproximability}. In the latter reduction, all distances are of the form $2^{-h}$ with exponent ranging from $0$ to $h=\cO\br{\log^{1-\epsilon} n}$. Therefore, we can scale all of them by a factor of $2^h=\text{poly}\br{n}$ to obtain integer and polynomially bounded distances without affecting the approximation ratio.

    Let $\br{T=(V, d), \mathcal{X}}$ be an instance of $\ProblemGST$ where $T$ is a tree metric with root $r$ and $\mathcal{X} \subseteq 2^V$ are the groups. We construct an instance of $\ProblemPCSC$ with outforest precedence constraints as follows. The universe is $\cX$ and the sets are created as follows: Since to include a node in the Steiner Tree we need to include its parent edge $e$ of length $d(e)$, for each vertex $v\neq r$ we create its $d(e)$ representatives $S_v^1 \preceq \ldots \preceq S_v^{d(e)}$ so that taking $S_v^{d(e)}$ to the cover requires taking all of the previous representatives. Additionally, for any directed edge $uv\in E$ such that $u \neq r$ and $e$ is the parent edge of $u$, we set $S^{d(e)}_u \preceq S^1_v$ to enforce the condition that including a node in a Steiner Tree requires including its parent node with its parent edge. If $v \in X$ and $e$ is the parent edge of $v$, we set $S_{d(e)}$ to cover $X$. It is easy to see that in order to cover any $X\in\cX$ by a vertex $v$, such that $x\in X$ one needs to include all representatives of edge $e$ above $v$. This means that any reasonable cover either picks no representative of a given vertex or all of them. By construction, any valid selection of sets in the constructed instance of $\ProblemPCSC$ corresponds to a valid Steiner Tree in the instance of $\ProblemGST$ and vice versa. Moreover, the cost of the selected sets is equal to the cost of the Steiner Tree. Since the construction size is polynomial, $\ProblemPCSC$ with outforest precedence constraints cannot be approximated within a factor of $\cO\br{\log^{2-\epsilon} n}$ for any $\epsilon > 0$ unless $\text{NP}\subseteq \text{ZTIME}\br{n^{\text{polylog}(n)}}$.
\end{proof}
As an immediate corollaries of Theorems \ref{thm:PCSC-to-PCWCDT-hardness}, \ref{thm:PCSC-outforest-hardness}, \ref{thm:PCMSSC-to-PCACDT-hardness} and \ref{thm:PCMSSC-outforest-hardness} we obtain the following results: 
\begin{corollary}\label{thm:PCWCDT-outforest-hardness}
    $\ProblemPCWCDT$ with outforest precedence constraints cannot be approximated within a factor of $\cO\br{\log^{2-\epsilon} n}$ for any $\epsilon > 0$ unless $\text{NP}\subseteq \text{ZTIME}\br{n^{\text{polylog}(n)}}$.
\end{corollary}

We now establish a similar result for $\ProblemPCMSSC$.
\begin{theorem}\label{thm:PCMSSC-outforest-hardness}
    $\ProblemPCMSSC$ with outforest precedence constraints cannot be approximated within a factor of $\cO\br{\log^{2-\epsilon} n}$ for any $\epsilon > 0$ unless $\text{NP}\subseteq \text{ZTIME}\br{n^{\text{polylog}(n)}}$.
\end{theorem}
\begin{proof}
    We show that $\ProblemGST$ on tree metrics is reducible to $\ProblemPCMSSC$ with outforest precedence constraints. Recall, that $\ProblemGST$ on trees cannot be approximated within a factor of $\cO\br{\log^{2-\epsilon} n}$ for any $\epsilon > 0$ unless $\text{NP}\subseteq \text{ZTIME}\br{n^{\text{polylog}(n)}}$ \cite{PolylogarithmicInapproximability}. By the same trick as above we may assume that all of the distances are integer and polynomially bounded.

    Let $\br{T=(V, d), \mathcal{X}}$ be an instance of $\ProblemGST$ where $T$ is a tree metric with a root $r$ and $\mathcal{X} \subseteq 2^V$ are the groups. We construct an instance of $\ProblemPCMSSC$ with outforest precedence constraints as follows. Let $L=\sum_{e\in T}d\br{e}$ be the total length of the tree. 
    The universe is $\cU=\cX\cup[L\cdot\spr{\cX}]$, we set $f=\frac{\spr{\cX}}{\spr{\cU}}$ and the sets are created as follows: Since to include a node in the Steiner Tree we need to include its parent edge $e$ of length $d(e)$, for each vertex $v\neq r$ we create its $d(e)$ representatives $S_v^1 \preceq \ldots \preceq S_v^{d(e)}$ so that taking $S_v^{d(e)}$ to the cover requires taking all of the previous representatives. Additionally, for any directed edge $uv\in E$ such that $u \neq r$ and $e$ is the parent edge of $u$, we set $S^{d(e)}_u \preceq S^1_v$ to enforce the condition that including a node in a Steiner Tree requires including its parent node with its parent edge. If $v \in X$ and $e$ is the parent edge of $v$, we set $S_{d(e)}$ to cover $X$. It is easy to see that in order to cover any $X\in\cX$ by a vertex $v$, such that $x\in X$ one needs to include all representatives of edge $e$ above $v$. This means that any reasonable cover either picks no representative of a given vertex or all of them. Moreover, we create a long chain of sets $S_1, S_2, \ldots, S_{L\cdot\spr{\cX}+1}$ such that $S_i \preceq S_{i+1}$ for all $i\in[L\cdot\spr{\cX}]$ and $S_{L\cdot\spr{\cX}+1}$ covers all of the elements in $[L\cdot\spr{\cX}]$. It is observe that any sensible cover never picks any of the sets $S_1, S_2, \ldots, S_{L\cdot\spr{\cX}+1}$. This is because by our choice of $f$, we only require to cover $\spr{\cX}$ elements, so even a naive solution $\sigma$ which picks representatives of all vertices in the tree (and thus covers entire~$\cX$) has cost bounded by $\coverageTime\br{\sigma}\leq L\cdot \spr{\cX}+L^2\cdot\spr{\cX}=L\cdot \spr{\cX}\cdot\br{L+1}$. On the other hand picking any of the sets $S_1, S_2, \ldots, S_{L\cdot\spr{\cX}}$ is reasonable only if we also pick $S_{L\cdot\spr{\cX}+1}$ which drives the cost of any such solution $\pi$ up to at least $\coverageTime\br{\pi}\geq\br{L\cdot\spr{\cX}+1}\cdot L\cdot\spr{\cX}$ which for $\spr{\cX}\geq 2$ is larger then $L\cdot \spr{\cX}\cdot\br{L+1}$. 
    
    Now, it is easy to see, that the contribution of elements in $\cX$ to the cost of the covering sequence is neglible by comparison to the contribution of elements in $[L\cdot\spr{\cX}]$. This is due to the fact that for any sequence $\sigma$, the contribution of elements in $\cX$ to $\coverageTime\br{\sigma}$ is at most $\spr{\cX}\cdot\spr{\sigma}\leq \spr{\cX}\cdot L$ while the contribution of elements in $[L\cdot\spr{\cX}]$ to $\coverageTime\br{\sigma}$ is at least $L\cdot\spr{\cX}\cdot\spr{\sigma}$. Threfore, asymptotically, minimizing $\coverageTime\br{\sigma}$ is equivalent to minimizing $\spr{\sigma}$. This in turn is equivalent to minimizing the size of the Steiner Tree (using a similar argument as in the previous reduction). Since the construction size is polynomial, $\ProblemPCMSSC$ with outforest precedence constraints cannot be approximated within a factor of $\cO\br{\log^{2-\epsilon} n}$ for any $\epsilon > 0$ unless $\text{NP}\subseteq \text{ZTIME}\br{n^{\text{polylog}(n)}}$.
\end{proof}
Importantly, because in the above reduction $f\neq 1$, this result cannot be easily transferred to $\ProblemPCACDT$.

\subsection{General precedence constraints}
In this section we show strong inapproximability results for $\ProblemPCWCDT$ and $\ProblemPCACDT$ with general precedence constraints by reducing from the \textsc{Planted Dense Subgraph Conjecture} which is a widely believed statement about hardness of detecting a dense component within an Erd\H{o}s-Renyi graph \cite{OnApproxTargetSetSelection,PCMSSC}.
\begin{conjecture}[Planted Dense Subgraph (PDS) Conjecture] \label{problem:PDS}
For any constants $\beta < \alpha$ and any $k\geq \sqrt{N}$, there is no polynomial time algorithm that can distinguish between the following two distributions of graphs with any advantage $\epsilon > 0$: 
\begin{itemize}
    \item With probability 1/2, an Erdős-Renyi graph $G(N, N^{\alpha - 1})$.
    \item With probability 1/2, an Erdős-Renyi graph $G(N, N^{\alpha - 1})$ with a planted subgraph of size $k$ and edge density $k^{\beta - 1}$.
\end{itemize}
\end{conjecture}
% The Planted Dense Subgraph Conjecture states that for any constants $ \beta < \alpha $ and any $k\geq \sqrt{N}$, there is no polynomial time algorithm that can distinguish between the following two distributions of graphs with any advantage $\epsilon > 0$:
% \begin{itemize}
%     \item With probability 1/2, $G_1$: an Erdős-Renyi graph $G(N, N^{\alpha - 1})$,
%     \item With probability 1/2, $G_2$: an Erdős-Renyi graph $G(N, N^{\alpha - 1})$ with a planted subgraph of size $k$ and edge density $k^{\beta - 1}$.
% \end{itemize}
Using this conjecture \cite{PCMSSC} have showed the following inapproximability bound:
\begin{theorem}\label{thm:PCMSSC-PDS-hardness}
    For any $\epsilon>0$, $\ProblemPCMSSC$ cannot be approximated within a factor of $\cO\br{m^{1/6-\epsilon}}$ nor $\cO\br{n^{1/12-\epsilon}}$ condition to \textsc{Planted Dense Subgraph Conjecture}.
\end{theorem}

By Theorem \ref{thm:PCMSSC-to-PCACDT-hardness} we immediately obtain:
\begin{corollary}\label{thm:PCACDT-PDS-hardness}
    For any $\epsilon>0$, $\ProblemPCACDT$ cannot be approximated within a factor of $\cO\br{m^{1/6-\epsilon}}$ nor $\cO\br{n^{1/12-\epsilon}}$ condition to \textsc{Planted Dense Subgraph Conjecture}.
\end{corollary}

Below we show that a similar reduction can also be used to obtain the same inapproximability result for $\ProblemPCSC$. The construction and analysis are almost the same as the one in \cite{PCMSSC} but the reduction is quite technical and we include it for completeness. In the reduction we will often make use of the Chernoff Bound which is as follows:
\begin{theorem}[Chernoff Bound]\label{thm:chernoff}
    Let $X_1, X_2, \ldots, X_n$ be independent random variables taking values in $\brc{0,1}$ such that for every $i\in[n]$, $\mathbb{P}[X_i=1]=p$. Let $X=\sum_{i=1}^n X_i$ and $\mu=\mathbb{E}[X]=p\cdot n$. Then, for every $\delta\in(0,1)$ it holds that:
    \begin{align*}
        \mathbb{P}[X\leq (1-\delta)\mu] &\leq e^{-\frac{\delta^2\mu}{2}},\\
        \mathbb{P}[X\geq (1+\delta)\mu] &\leq e^{-\frac{\delta^2\mu}{3}}.
    \end{align*}
\end{theorem}

We have the following theorem:
\begin{theorem}\label{thm:PCSC-PDS-hardness}
    For any $\epsilon>0$, $\ProblemPCSC$ cannot be approximated within a factor of $\cO\br{m^{1/6-\epsilon}}$ nor $\cO\br{n^{1/12-\epsilon}}$ condition to \textsc{Planted Dense Subgraph Conjecture}.
\end{theorem}
\begin{proof}
    \par\noindent\textbf{Construction of the reduction.} We start by choosing appropriate parameters for our reduction. Let $k=\sqrt{N}$, $\alpha=1/2$ and $\beta=1/2-\gamma$ for some $\gamma > 0$ to be determined later. In our construction we will embed the structure of the input graph into the precedence constraints, while the universe elements will merely enforce the covering requirement.
    
    Given a graph $G=([N], \mathcal{E})$ as input to the \textsc{Planted Dense Subgraph Conjecture}, we construct an instance of $\ProblemPCSC$ as follows. For each vertex $v \in [N]$ of $G$, we create $\lambda$ representative sets $V_{v, i}$ for $i \in [\lambda]$, where $\lambda$ is an appropriately chosen natural number to be specified later. For each edge $uv\in\mathcal{E}$, we create one representative set $E_{uv}$ which must be preceded by all representatives of both $u$ and $v$. That is, for all $i \in [\lambda]$, we have $V_{u, i} \preceq E_{uv}$ and $V_{v, i} \preceq E_{uv}$.

    Next, we define the universe of elements to be covered. Let $\cU=\{0,1,\dots, n\}$ for some appropriately chosen natural number $n$ to be determined later. Every vertex representative set contains only the special element $0$, that is, $V_{v, i}=\{0\}$ for all $v \in [N]$ and $i \in [\lambda]$.
    
    To define the contents of edge representative sets, we use a randomized construction based on a parameter $p\in(0,1)$. For each vertex $v \in [N]$, we define an auxiliary set $\cU_v \subseteq [n]$ by including each element $j \in [n]$ in $\cU_v$ independently with probability $p$. That is, for each $j \in [n]$ and each $v \in [N]$, we have $\mathbb{P}[j \in\cU_v] = p$. The edge representative set $E_{uv}$ is then defined to cover the intersection of the auxiliary sets, that is, $E_{uv} = \cU_u \cap \cU_v$.
    
    From this construction, we can compute the following expected values:
    \begin{align*}
        \mathbb{E}[\spr{\cU_v}] &= p\cdot n &&\text{for each vertex } v,\\
        \mathbb{E}[\spr{\{v\colon j\in\cU_v\}}] &= p\cdot N &&\text{for each element } j \in [n],\\
        \mathbb{E}[\spr{E_{uv}}] &= \mathbb{E}[\spr{\cU_u \cap \cU_v}] = n\cdot p^2 &&\text{for each edge } uv.
    \end{align*}

    The intuition behind this construction is as follows: we choose $p$ to be the smallest value such that if a planted component exists, then with high probability it covers all of $\cU$ using only edge representatives from the planted component (and their required predecessors). On the other hand, if no planted component exists, then we show that any solution must use significantly more sets to cover $\cU$, thus establishing the desired inapproximability ratio. 
    
    \par\noindent\textbf{Choice of parameters.}
    Let $\mathcal{P}$ denote the planted component if it exists. Since $\mathcal{P}$ consists of $k = \sqrt{N}$ vertices, for any element $j \in [n]$, the expected number of vertices in $\mathcal{P}$ whose auxiliary sets contain $j$ is 
    \[
    \mathbb{E}[\spr{\{v\colon j\in\cU_v, v \in V(\mathcal{P})\}}] = p\cdot\sqrt{N}.
    \] 
    Therefore, the expected number of unordered pairs of vertices in $\mathcal{P}$ where both auxiliary sets contain $j$ is:
    \[
    \mathbb{E}\left[\left|\left\{(v, u)\colon j\in\cU_v, j\in\cU_u, v,u \in V(\mathcal{P}), v < u\right\}\right|\right] = \binom{p\cdot\sqrt{N}}{2}.
    \]
    
    Since each edge in $\mathcal{P}$ exists independently with probability 
    \[
    k^{\beta - 1} = (\sqrt{N})^{-1/2 - \gamma}=N^{-1/4-\gamma/2},
    \]
    the expected number of edge representative sets from $\mathcal{P}$ that cover element $j$ is:
    \[
    \mathbb{E}[\spr{\{E_{uv}\colon j\in E_{uv}, uv \in E(\mathcal{P})\}}] = \binom{p\cdot\sqrt{N}}{2} \cdot N^{-1/4-\gamma/2}\geq \frac{p^2\cdot N}{4} \cdot N^{-1/4-\gamma/2}=\frac{p^2\cdot N^{3/4-\gamma/2}}{4}.
    \] 
    For this expected value to be sufficiently large (specifically, to be $\Theta(\log^2 N)$), we choose:
    $
    p=32\cdot N^{-3/8+\gamma/4}\cdot \log N
    $. 
	We have 
    \[
    \mathbb{E}[\spr{\{E_{uv}\colon j\in E_{uv}, uv \in E(\mathcal{P})\}}] \geq 64\cdot\log^2 N,
    \]
    which will allow us to apply concentration bounds effectively.

    \par\noindent\textbf{Case 1: Planted component exists.} Assume that the input graph contains a planted component $\mathcal{P}$. We show that there exists a solution to $\ProblemPCSC$ using only sets corresponding to vertices and edges in $\mathcal{P}$ that covers all elements in $\cU$ with high probability, and that this solution has cost $O(\lambda\cdot\sqrt{N}+N^{3/4-\gamma/2})$.
    Consider the solution that selects all vertex representatives $V_{v,i}$ for all $v \in V(\mathcal{P})$ and all $i \in [\lambda]$, together with all edge representatives $E_{uv}$ for all edges $uv \in E(\mathcal{P})$. The special element $0$ is covered by any vertex representative, so we focus on covering elements in $[n]$.
    
    Fix an arbitrary element $j \in [n]$. We first bound the number of vertices in $\mathcal{P}$ whose auxiliary sets contain $j$. Since 
    \[
    \mathbb{E}[\spr{\{v\colon j\in\cU_v, v \in V(\mathcal{P})\}}] = p\cdot\sqrt{N},
    \] 
    by Chernoff's bound (Theorem \ref{thm:chernoff}) with $\delta = 1/2$, we have:
    \[
    \mathbb{P}\left[\left|\{v\colon j\in\cU_v, v \in V(\mathcal{P})\}\right| \leq \frac{p\cdot\sqrt{N}}{2}\right] \leq  e^{-p\cdot\sqrt{N}/8} = e^{-4\cdot N^{1/8+\gamma/4}\cdot \log N} = N^{-4\cdot N^{1/8+\gamma/4}},
    \]
    which is exponentially small. 
    
    Next, for large enough $x$, we have $\binom{x}{2} \geq x^2/4$. Therefore, the expected number of unordered pairs of vertices in $\mathcal{P}$ where both auxiliary sets contain $j$ is at least $p^2\cdot N/4$. Again by Chernoff's bound, we have:
    \[
    \mathbb{P}\left[\left|\{(u,v)\colon j\in\cU_u, j\in\cU_v, u,v \in V(\mathcal{P}), u < v\}\right| \leq \frac{p^2\cdot N}{16}\right] \leq N^{-4\cdot N^{1/8+\gamma/4}}.
    \]

    Now, each such pair $(u,v)$ corresponds to an edge in $\mathcal{P}$ with probability $N^{-1/4-\gamma/2}$. Since edge existence is independent of the auxiliary set memberships, the expected number of edge representatives in $\mathcal{P}$ that cover $j$ is:
    \[
    \mu := \mathbb{E}[\spr{\{E_{uv}\colon j\in E_{uv}, uv\in E(\mathcal{P})\}}] \geq \frac{p^2\cdot N}{16} \cdot N^{-1/4-\gamma/2}=64\cdot\log^2 N.
    \]

    Applying Chernoff's bound once more, we obtain:
    \[
    \mathbb{P}[\spr{\{E_{uv}\colon j\in E_{uv}, uv\in E(\mathcal{P})\}} \leq \mu/2] \leq e^{-\mu/8} = e^{-8\log^2 N} = N^{-8\log N}.
    \]

    By the union bound over all $n$ elements in $[n]$, the probability that some element is not covered by the edge representatives in $\mathcal{P}$ is at most $n \cdot N^{-8\log N}$. We will choose $n = \Theta(N^{3/4})$ later, which makes this probability negligible for large $N$.
    
    Finally, we bound the number of sets in this solution. The number of vertex representative sets is exactly $\lambda\cdot\sqrt{N}$. For the edge representatives, observe that:
    \[
    \mathbb{E}[\spr{E(\mathcal{P})}]=\binom{\sqrt{N}}{2}\cdot N^{-1/4-\gamma/2}\leq \frac{N}{2} \cdot N^{-1/4-\gamma/2} = \frac{N^{3/4-\gamma/2}}{2}.
    \]
    
    By Chernoff's bound with $\delta = 1$, we have:
    \[
    \mathbb{P}[\spr{E(\mathcal{P})} \geq N^{3/4-\gamma/2}] \leq e^{-N^{3/4-\gamma/2}/6},
    \]
    which is extremely small. Therefore, with high probability, the total number of sets in the solution is at most:
    \[
    \lambda\cdot\sqrt{N}+N^{3/4-\gamma/2}.
    \]
    
    \par\noindent\textbf{Case 2: No planted component.} Now assume that the input graph does not contain a planted component, so it is drawn entirely from $G(N, N^{-1/2})$. We show that any solution to $\ProblemPCSC$ must use significantly more sets to cover all elements in $\cU$ with high probability.
    
    Consider any solution to $\ProblemPCSC$ that selects representatives corresponding to some set $X \subseteq [N]$ of vertices with $|X| = x$. Due to the precedence constraints, this solution must include all $\lambda \cdot x$ vertex representatives $V_{v,i}$ for $v \in X$ and $i \in [\lambda]$. After selecting these vertex representatives, the solution can then select edge representatives $E_{uv}$ for edges $uv$ where both $u, v \in X$.
    
    We first bound the number of edges among the vertices in $X$. Since the graph is drawn from $G(N, N^{-1/2})$, each potential edge exists independently with probability $N^{-1/2}$. For any fixed set $X$ of size $x$, the expected number of edges with both endpoints in $X$ is:
    \[
    \mu_E := \binom{x}{2}\cdot N^{-1/2} \leq \frac{x^2}{2\cdot\sqrt{N}}.
    \]
    
    There are at most $\binom{N}{x} \leq N^x$ possible sets $X$ of size $x$. By Chernoff's bound with $\delta = 1/2$, the probability that a specific set $X$ contains more than $3\cdot\mu_E/2$ edges is at most $e^{-\mu_E/12}$. By the union bound over all possible sets of size $x$, the probability that some set of size $x$ has more than $3\cdot\mu_E/2$ edges is at most:
    \[
    N^x \cdot e^{-\mu_E/12} = N^x \cdot \exp\left(-\frac{x^2}{24\cdot\sqrt{N}}\right).
    \]
    
    For this probability to be at most $N^{-8}$, we would need $\mu_E \geq 12\cdot(x+8)\cdot\log N$, which translates to $x \geq 500\cdot\sqrt{N}\cdot\log N$. However, we require $x \leq N^{(5-2\gamma)/8}$. For such values of $x$ and sufficiently large $N$ (note that $N^{(5-2\gamma)/8} < N^{5/8} < \sqrt{N}$ for large $N$), we can safely assume that with overwhelming probability, any set $X$ of size $x \leq N^{(5-2\gamma)/8}$ contains at most $x^2/\sqrt{N}$ edges.
    
    Next, we bound how many elements can be covered by these edge representatives. By construction, each edge representative $E_{uv}$ covers the elements in $(\cU_u \cap \cU_v)$. The size of $\cU_u \cap \cU_v$ has expectation $\mathbb{E}[\spr{\cU_u \cap \cU_v}]=np^2$. By Chernoff's bound, the probability that a single edge representative covers more than $2np^2$ elements from $[n]$ is at most $e^{-np^2/12}$. 
    
    The total number of edge representatives available after selecting vertices in $X$ is at most $x^2/\sqrt{N}$. By the union bound, the probability that at least one of these edge representatives covers more than $2n\cdot p^2$ elements from $[n]$ is at most:
    \[
    \frac{x^2}{\sqrt{N}} \cdot e^{-n\cdot p^2/12}.
    \]
    
    Recall that 
    \[
    p^2 = (32)^2 \cdot N^{-3/4+\gamma/2}\cdot \log^2 N = 1024 \cdot N^{-3/4+\gamma/2}\cdot \log^2 N.
    \] 
    For the above probability to be negligible, we need $n \geq 200\cdot\log N/p^2$, which gives:
    \[
    n \geq \frac{200\cdot \log N}{1024 \cdot N^{-3/4+\gamma/2}\cdot \log^2 N} = \cO\left(\frac{N^{3/4-\gamma/2}}{\log N}\right).
    \]
    This condition is satisfied by choosing $n = \Theta(N^{3/4})$.
    
    Now, suppose $x \leq N^{(5-2\gamma)/8}$. Then, with high probability, after selecting $\lambda \cdot x$ vertex representatives and at most $x^2/\sqrt{N}$ edge representatives, the total number of elements from $[n]$ that are covered is at most:
    \begin{align*}
    2np^2 \cdot \frac{x^2}{\sqrt{N}} 
	&= 
    2048 \cdot n \cdot N^{-3/4+\gamma/2}\cdot \log^2 N \cdot \frac{x^2}{\sqrt{N}} = 2048 \cdot n \cdot N^{-5/4+\gamma/2}\cdot \log^2 N \cdot x^2
	\\&\leq 2048 \cdot n \cdot N^{-5/4+\gamma/2}\cdot \log^2 N \cdot N^{(5-2\gamma)/4} = 2048 \cdot n \cdot \log^2 N.
    \end{align*}
    
    Since $n = \Theta(N^{3/4})$, we have $2048 \cdot n \cdot \log^2 N = o(n)$ for sufficiently large $N$. This means that when $x \leq N^{(5-2\gamma)/8}$, at most $o(n)$ elements from $[n]$ are covered, leaving $\Omega(n)$ elements uncovered.
    
    To cover all elements in $\cU$, any solution must eventually select enough sets to cover these remaining $\Omega(n)$ elements. However, due to the precedence constraints, to reach edge representatives that can cover many new elements, the solution must first select many vertex representatives. Specifically, to significantly increase coverage beyond what $x \leq N^{(5-2\gamma)/8}$ vertices provide, the solution must select at least $\lambda\cdot N^{(5-2\gamma)/8}$ vertex representatives plus at least $N^{(5-2\gamma)/4}$ edge representatives. Therefore, with high probability, the cost of any solution that covers all elements in $\cU$ is at least:
    \[
    \lambda\cdot N^{(5-2\gamma)/8} + N^{(5-2\gamma)/4}.
    \]
    
    \par\noindent\textbf{Comparison and conclusion.}
    We now compare the costs in the two cases to establish the inapproximability ratio. By choosing $\lambda = N^{1/4}$ (or any $\lambda \geq N^{1/4}$) and letting $\gamma$ approach $0$ we have that the ratio between the costs in Case 2 and Case 1 is at least:
    \[
    \frac{\lambda\cdot N^{(5-2\gamma)/8} + N^{(5-2\gamma)/4}}{\lambda\cdot\sqrt{N} + N^{3/4-\gamma/2}} = \Omega\br{
    \frac{N^{5/8}}{N^{1/2}}} = \Omega\br{N^{1/8}}.
    \]

    Since we chose $n = \Theta(N^{3/4})$, we have that the gap is $\Theta(N^{1/8}) = \Theta((N^{3/4})^{1/6}) = \Theta(n^{1/6})$. By choosing $\gamma$ sufficiently small, we can make this gap arbitrarily close to $n^{1/6}$.
    Furthermore, the total number of sets $m$ in the constructed instance is $m = \lambda\cdot N + |\mathcal{E}|$. For a graph drawn from $G(N, N^{-1/2})$, the number of edges is highly concentrated around $\binom{N}{2} \cdot N^{-1/2} = \Theta(N^{3/2})$. With $\lambda = N^{1/4}$, we have $m = \Theta(N^{3/2})$. Therefore, $n = \Theta(N^{3/4}) = \Theta(m^{1/2})$, and the gap can also be expressed as:
    $
    \Theta(N^{1/8}) = \Theta((N^{3/2})^{1/12}) = \Theta(m^{1/12}).
    $
    
    Thus, any polynomial-time algorithm for $\ProblemPCSC$ with approximation factor better than $\cO(n^{1/6-\epsilon})$ or $\cO(m^{1/12-\epsilon})$ for any $\epsilon > 0$ could distinguish between the two cases of the Planted Dense Subgraph problem with significant advantage. This completes the proof.
\end{proof}
Therefore, we immediately get that:
\begin{corollary}\label{thm:PCWCDT-PDS-hardness}
    For any $\epsilon>0$, $\ProblemPCWCDT$ cannot be approximated within a factor of $\cO\br{m^{1/6-\epsilon}}$ nor $\cO\br{n^{1/12-\epsilon}}$ condition to \textsc{Planted Dense Subgraph Conjecture}.
\end{corollary}

\section{Conclusions and open problems}

We have provided a series of approximation algorithms directed towards solving decision tree and set covering problems subject to precedence constraints. Our main contribution is a series of reductions which enables us to systematically relate the two types of problems independently of the specific structure of the precedence constraints. As an immediate corollary, we have obtained both general $\tilde{\cO}\br{\sqrt{m}}$-approximation algorithms for both types of problems as well as case-specific, tailored polylogarithmic approximations for both inforest and outforest precedence constraints. These two classes are natural to study, since trees are among most prevalent structures encountered in practice.

Prior to this work, almost no non-trivial approximation algorithms were known for either of these problems, even for the simplest of our variants. Among possible explanations for this state of affairs is the fact that achieving our results demanded the development of an elaborate and systematic reduction framework and usage of non-greedy procedures. This is particularly visible for the average-case criterion where the algorithm for the $\ProblemPCMSSC$ required usage of an elaborate procedure, which used a greedy algorithm as a subroutine, however itself was not greedy. This provokes two types of questions: Firstly, is it possible to show that an algorithm, which always picks a test (with its closure) maximizing some kind of information gain also achieves a good approximation ratio for either $\ProblemPCWCDT$ or $\ProblemPCACDT$? Secondly, is it possible to simplify the procedure for $\ProblemPCACDT$ along with its analysis while still maintaining quality approximation ratio?

We have complemented our algorithmic results with a series of matching lower bounds, providing inapproximability results for most of our variants including both outforest and general precedence constraints. Of particular surprise is to us the fact that approximating the decision tree problem generalizes the task of planted dense subgraph detection in a random graph. This result, although quite unexpected, extends the lengthy line of work regarding connection between average-case complexity and inapproximability of problems (to read more about such connections see \cite{Feige2002RelationsBetweenAverageCaseComplexityAndApproximationComplexity}). Whether or not polynomial factor inapproximability results based on P vs NP can be established for our problem remains an open question.

Hereby, we have studied average-case decision tree construction under the assumption that the probability distribution over all possible hypothesis is uniform. In particular, extending our results to non-uniform distributions would require new techniques. It seems to us that that our analysis relies heavily on the uniformity assumption. However, it is possible that some of the techniques we have developed can be adapted to handle non-uniform distributions as well. For example our algorithms for covering problems should be easily adaptable to include non-uniform weights on the elements of the universe, which might be useful for encoding non-uniform distributions over the hypotheses. However, it is not clear to us whether we can adapt the reduction from $\ProblemPCACDT$ to $\ProblemPCMSSC$ to work for that setup as well. We additionally remark, that addition of probability distribution would allow us to establish an inapproximability result for $\ProblemPCACDT$ with outforest precedence constraints. To see this, one can observe, that we can use a reduction similar to the one in Theorem~\ref{thm:PCMSSC-outforest-hardness}. The only important difference is that in order for the reduction to work, we give a small but non-zero probability to all elements in $\cX$ and instead of adding a large number of dummy elements to force the algorithm to optimize for the length of the covering sequence, we can add one element with a very high probability. This way, the algorithm will be forced to optimize for the average case, which will be dominated by the cost of identifying the high probability element. By replacing sets with accordingly created tests (see the reduction in Theorem~\ref{thm:PCMSSC-to-PCACDT-hardness}), we get that the $\ProblemPCACDT$ with outforest precedence constraints and non-uniform probabilities cannot be approximated within $\cO(\log^{2-\epsilon} n)$ factor for any $\epsilon\geq 0$, unless $\text{NP}\subseteq \text{ZTIME}(n^{\text{polylog}(n)})$.

A different kind of further direction is to enrich the \textsc{Optimal Decision Tree} with more types of constraints borrowed from scheduling theory. For example one may consider release dates on the tests or due dates on the elements. In such a setup it is also possible to optimize for different cost functions, such as the maximum lateness or the total tardiness. The same types of generalizations are also applicable to the set covering problems. To the best of our knowledge, no results of this nature are known for any of these problems. 

\bibliographystyle{plain}
\bibliography{bib-pcal}

\end{document}